\documentclass[11pt]{article}
\usepackage[margin=1in]{geometry}
\usepackage{setspace}
\usepackage[english]{babel}
\usepackage{xr}

\usepackage{color}

\usepackage{mathrsfs,hyperref, algorithm}
\usepackage[]{amsmath}
\usepackage{amsthm}
\usepackage{fix-cm}
\usepackage[]{amssymb}
\usepackage[]{latexsym}
\usepackage[right]{eurosym}
\usepackage[T1]{fontenc}
\usepackage[]{graphicx}
\usepackage{fancyhdr}
\usepackage{enumitem}
\setlist[itemize]{wide,labelwidth=!,labelindent=0pt}
\setlist[enumerate]{wide,labelwidth=!,labelindent=0pt}

\allowdisplaybreaks

\makeatletter

\newcommand{\bbr}{\mathbb{R}}
\newcommand{\E}{\mathbb{E}}
\newcommand{\W}{\mathcal{W}}

\newcommand{\bbn}{\mathbb{N}}
\renewcommand{\P}{\mathbb{P}}

\newcommand{\bbq}{\mathbb{Q}}
\newcommand{\Q}{\bbq}

\newcommand{\bbt}{\mathbb{T}}

\newcommand{\fcal}{\mathcal{F}}
\newcommand{\gcal}{\mathcal{G}}
\newcommand{\dcal}{\mathcal{D}}

\newcommand{\rcal}{\mathcal{R}}

\newcommand{\acal}{\mathcal{A}}

\newcommand{\ucal}{\mathcal{U}}

\newcommand{\M}{{\cal M}}
\newcommand{\D}{{\cal D}}

\newcounter{modcount}
\newcommand{\modulo}[2]{%
\setcounter{modcount}{#1}\relax
\ifnum\value{modcount}<#2\relax
\else\relax
\addtocounter{modcount}{-#2}\relax
\modulo{\value{modcount}}{#2}\relax
\fi}
\newcommand{\tablepictures}[4][c]{\begin{tabular}[#1]{@{}c@{}}#2\vspace{0.5cm}\\(\alph{#4}) #3\end{tabular}}
\newcounter{gridsearch}
\newcommand{\tabpic}[2]{
    \stepcounter{gridsearch}
    \modulo{\thegridsearch}{2}
    \ifnum\value{modcount}=0
        \tablepictures[t]{#1}{#2}{gridsearch}\\[2.0cm]
    \else
        \tablepictures[t]{#1}{#2}{gridsearch}&~&
    \fi
}

\makeatother
\newtheorem{lemma}{Lemma}[section]
\newtheorem{proposition}[lemma]{Proposition}
\newtheorem{theorem}[lemma]{Theorem}
\newtheorem{corollary}[lemma]{Corollary}
\newtheorem{definition}[lemma]{Definition}
\newtheorem{example1}[lemma]{Example}
\newtheorem{rem1}[lemma]{Remark}

\newtheorem{alg1}[lemma]{Algorithm}
\newtheorem{me1}[lemma]{Mechanism}

\newenvironment{remark}{\begin{rem1}\rm}{\end{rem1}}

\newcommand{\T}{\top}
\newcommand{\diag}{\operatorname{diag}}

\DeclareMathOperator*{\essinf}{ess\,inf}
\DeclareMathOperator*{\esssup}{ess\,sup}

\newcommand\indn[1]{\mathbb{I}_{#1}}
\newcommand\ind[1]{\indn{\{#1\}}}
\newcommand\sbullet{\scalebox{0.5}{$\bullet$}}

\begin{document}
\title{Set-Valued Dynamic Risk Measures for Processes and Vectors}
\author{
Yanhong Chen \thanks{Hunan University, College of Finance and Statistics, Changsha, China 410082. Research is partially supported by National Natural Science Foundation of China (No.\ 11901184) and Natural Science Foundation of Hunan Province (No.\ 2020JJ5025).}
\and
Zachary Feinstein \thanks{Stevens Institute of Technology, School of Business, Hoboken, NJ 07030, USA. \tt{zfeinste@stevens.edu}}}
\date{\today}
\maketitle
\begin{abstract}
The relationship between set-valued risk measures for processes and vectors on the optional filtration is investigated.
The equivalence of risk measures for processes and vectors and the equivalence of their penalty function formulations are provided.
In contrast with scalar risk measures, this equivalence requires an augmentation of the set-valued risk measures for processes.
We utilize this result to deduce a new dual representation for risk measures for processes in the set-valued framework.
Finally, the equivalence of multiportfolio time consistency between set-valued risk measures for processes and vectors is provided; to accomplish this, an augmented definition for multiportfolio time consistency of set-valued risk measures for processes is proposed.
\end{abstract}

\section{Introduction}\label{sec:intro}

Coherent risk measures for multi-period models were introduced by Artzner et al.~\cite{AD02,AD07},
where time consistency (i.e., risk preferences are consistent over time) of risk measures was shown to be equivalent to Bellman's principle (i.e., the risk of a portfolio is equivalent to the iterated risks backwards in time).
These risk measures allow for the risk of cash flows to be quantified in an axiomatic framework.  Though these risk measures appear distinct from the more traditional risk measures of, e.g., Artzner et al.~\cite{AD97,AD99}, F\"{o}llmer and Schied~\cite{FS02} and Frittelli and Rosazza Gianin~\cite{FG02}, risk measures applied to processes can be proven to be equivalent to these more traditional risk measures via the optional filtration as demonstrated in, e.g., Acciaio et al.~\cite{acciaio2012risk}.

Within this work, we are concerned with dynamic risk measures.  That is, we consider those risk measures which allow for updates to the minimal capital requirement over time due to revealed information as encoded in a filtration.
First, there are many studies about dynamic risk measures for random variables which describe financial positions, for instance:
Bion-Nadal~\cite{BN04} studied the dual or robust representation;
Bion-Nadal~\cite{BN08,BN09} gave the equivalent characterization of time consistency of dynamic risk measures for random variables by a cocycle condition on the minimal penalty function (i.e., such that the penalty functions can be decomposed as a summation over time);
Barrieu and El Karoui~\cite{BE-09} studied some applications of static and dynamic risk measures in the aspect of pricing, hedging and designing derivative; and
Delbaen et al.~\cite{DPRG10} studied the representation of the penalty function of dynamic risk measures for random variables induced by backward stochastic differential equations.
There also exist many studies about dynamic risk measures applied to cash flows, for example:
Riedel~\cite{R04} studied the dual representation for time consistent dynamic coherent risk measures for processes with real values;
Cheridito et al.~\cite{CDK04,CDK05} studied dynamic risk measures on the space of all bounded and unbounded c\`{a}dl\`{a}g processes, respectively, with real values or generalized real values;
Frittelli and Scandolo~\cite{FS06} proposed dynamic risk measures for processes from a perspective of acceptance sets;
Cheridito et al.~\cite{CDK06} studied dynamic risk measures for bounded discrete-time processes with random variables as values,
and provided a dual representation and an equivalent characterization of time consistency in terms of the additivity property of the corresponding acceptance sets; and
Acciaio et al.~\cite{acciaio2012risk} provided supermartingale characterisation (i.e., that the sum of a risk measure and its minimal penalty function form a supermartingale w.r.t.\ the dual probability measure) of time consistency of dynamic risk measures for processes.
As highlighted in~\cite{acciaio2012risk}, as with the static risk measures, these dynamic risk measures for processes can be found to be equivalent to risk measures over random variables through the use of the optional filtration.  Utilizing the optional $\sigma$-algebra, the equivalence of time consistency for traditional dynamic risk measures and those for processes are also found to be equivalent.

Set-valued risk measures were introduced in Jouini et al.~\cite{JMT04}, Hamel~\cite{H09}, Hamel et al.~\cite{HHR10} and extended to the dynamic framework in Feinstein and Rudloff~\cite{FR12,FR12b}.  All of those works present risk measures for random vectors.  Set-valued risk measures were introduced so as to consider risks in markets with frictions. Such risk measures have, more recently, been utilized for quantifying systemic risk in Feinstein et al.~\cite{FRW15}
and Ararat and Rudloff~\cite{ararat2020dual}.
In Chen and Hu~\cite{CH17,CH20}, set-valued risk measures for processes were introduced.  However, as demonstrated in those works, the equivalence of risk measures for processes and risk measures for random vectors no longer holds except under certain strong assumptions.  These strong assumptions limit the direct application of known results for, e.g., multiportfolio time consistency (i.e., a version of Bellman's principle for set-valued risk measures) as proven in Feinstein and Rudloff~\cite{FR15-supermtg,FR18-tc} to set-valued risk measures for processes.

In this work we prove a new relation between set-valued risk measures for processes and those for random vectors (on the optional filtration).
Notably, the formulation for equivalence indicates that risk measures for processes at time $t$ need to be augmented to capture the risks prior to time $t$ due to the non-existence of a unique ``0'' element in the set-valued framework; such an augmentation is not necessary in the scalar case as in, e.g.,~\cite{acciaio2012risk}.
By utilizing this new equivalence relation, we are able to derive a novel dual representation for set-valued risk measures for processes akin to that done in~\cite{acciaio2012risk}.
Additionally, we extend these results to present the equivalent formulation for multiportfolio time consistency; since the equivalence between risk measure settings requires an augmentation, we present a new definition for time consistency for risk measures for processes.

The organization of this paper is as follows.  In Section~\ref{sec:rm}, we present preliminary and background information on the set-valued risk measures of interest within this work, i.e., as functions of stochastic processes and random vectors.  Section~\ref{sec:equiv} demonstrates the equivalence of these risk measures both in primal and dual representations.  With the focus on dynamic risk measures, the equivalence of multiportfolio time consistency for processes and vectors is presented in Section~\ref{sec:mptc}.  Section~\ref{sec:conclusion} concludes.

\section{Set-valued risk measures}\label{sec:rm}
In this section we will summarize the definitions for set-valued risk measures.
We will provide details of the filtrations and the spaces of claims of interest within this work in Section~\ref{sec:rm-prelim}, and give the necessary background on risk measures for processes as defined in, e.g.,~\cite{CH17,CH20}, and for random vectors as defined in, e.g.,~\cite{FR12,FR12b} in Section~\ref{sec:rm-process} and Section~\ref{sec:rm-vector}, respectively.

\subsection{Background and notation}\label{sec:rm-prelim}

Fix a finite time horizon $T \in \bbn:=\{1, 2, \dots\}$ and denote ${\bbt}:=\{0, 1, \dots, T\}$ and
 ${\bbt_t}:=\{s \in \bbt  \; : \; s \geq t\}$ for $t \in \bbt$.
Furthermore, define the discrete interval $[t, s) := \{t, t+1, \ldots, s-1\}$ for any $s, t \in {\bbt}$ with $s > t$.
Fix a filtered probability space $(\Omega, \fcal, (\fcal_t)_{t \in \bbt}, \P)$
with $\fcal_0$ the complete and trivial $\sigma$-algebra.  Without loss of generality, take $\fcal=\fcal_T$.
Let $d \geq 1$ be the number of assets under consideration and $|\cdot|$ be an arbitrary but fixed norm on $\bbr^d$.
Let $L_t^0(\bbr^d) := L^0(\Omega, \fcal_t, \P; \bbr^d)$ (with $L^0(\bbr^d) := L_T^0(\bbr^d)$) denote the linear space of the equivalence classes of $\fcal_t$-measurable vector-valued functions, where we specify random vectors $\P$-a.s.
Define the space of (equivalence classes of) $p$-integrable random $d$-vectors for $p \in \{1,\infty\}$ by $L_t^p(\bbr^d) :=L^p(\Omega, \fcal_t, \P; \bbr^d) \subseteq L_t^0(\bbr^d)$ (with $L^p(\bbr^d) := L_T^p(\bbr^d)$).  In this way, $L_t^p(\bbr^d)$ denotes the linear space of the equivalence classes of $\fcal_t$-measurable functions $X: \Omega \to \bbr^d$ with bounded $p$ norm where these norms are given by
\begin{align*}
\|X\|_p = \begin{cases} \E^\P[|X|]=\int\limits_{\Omega} |X(\omega)| d \P &\text{if } p=1 \\ \esssup_{\omega \in \Omega}|X(\omega)| &\text{if } p=\infty.\end{cases}
\end{align*}
In this paper, we denote the expectation $\E^\P[X]$ by $\E[X]$, and denote the conditional expectation $\E^\P[X | \fcal_t]$ by $\E_t[X]$.
Throughout this work,  we consider the weak* topology on $L_t^\infty(\bbr^d)$ such that the dual space of $L_t^\infty(\bbr^d)$ is $L_t^1(\bbr^d)$.

Fix $m \in \{1, \dots, d\}$ of the assets to be eligible for covering the risk of a portfolio.
Denote by $M := \bbr^{m} \times \{0\}^{d-m}$ the subspace of eligible assets (those assets which can be used to satisfy capital requirements, e.g., US dollars and Euros).
For any measurable set $B \subseteq \bbr^d$,
write $L_t^{p}(B):=\{X \in L_t^{p}(\bbr^d) \; : \; X \in B \, \P\mbox{-a.s.}\}$ for those random vectors that take $\P$-a.s.\ values in
 $B$ for $p \in \{0,1,\infty\}$.
In particular, we denote the closed convex cone of uniformly bounded $\bbr^d$-valued $\fcal_t$-measurable random vectors with $\P$-a.s.\ non-negative components by $L_t^\infty(\bbr^d_+)$, and let
$L_t^{\infty}(\bbr^d_{++})$ be the $\fcal_t$-measurable random vectors with $\P$-a.s.\ strictly positive components.
Additionally, we denote $M_t:= L_t^{\infty}(M)$ to be the space of time $t$ measurable eligible portfolios.  Throughout we consider the non-negative eligible portfolios $M_{t, +}:= M_t \cap L_t^{\infty}(\bbr^d_+)$,
the positive dual cone $M_{t, +}^{\ast}:=\{u \in L_t^1(\bbr^d) \; :  \; \E[u^\T m] \geq 0 \text{ for any } m \in M_{t, +}\}$ and
the perpendicular space $M_{t, +}^{\perp}:=\{u \in L_t^1(\bbr^d) \;  : \; \E[u^\T m] = 0 \text{ for any } m \in M_{t, +}\}$.
Generally, we will define $C^{\ast}:=\{u \in L_t^1(\bbr^d) \; : \; \E[u^\T m] \geq 0 \text{ for any } m \in C\}$ for any convex cone $C \subseteq L_t^{\infty}(\bbr^d)$.
(In)equalities between random vectors, stochastic processes
and between sets are always understood componentwise in the $\P$-a.s.\ sense, unless stated otherwise.
The multiplication between a random variable $\lambda \in L_t^\infty(\bbr)$ and a set of random vectors
$B_t \subseteq L_t^{\infty}(\bbr^d)$ is understood as in the elementwise sense, i.e.,
$\lambda B_t :=\{\lambda X \; : \; X \in B_t\}$ with $(\lambda X)(\omega)=\lambda(\omega) X(\omega)$.
Given a set $A \subseteq M_t$, $\mbox{cl}(A)$ means the closure of $A$ w.r.t.\ the subspace topology on $M_t$
and $\mbox{co}(A)$ means the convex hull of $A$.
Denote the spaces of upper and closed-convex upper sets respectively by
\begin{align*}
\ucal(M_t;M_{t,+}) &:= \{D \subseteq M_t \;  :  \; D=D+ M_{t, +}\}\\
\gcal(M_t;M_{t,+}) &:= \{D \subseteq M_t \; : \; D=\mbox{cl co}(D+ M_{t,+})\},
\end{align*}
where the $+$ sign denotes the usual Minkowski addition. 
We will denote $\Gamma_t(w)$ with $w \in L_t^1(\bbr^d) \setminus \{0\}$ by
$\Gamma_t(w) = \{u \in L_t^{\infty}(\bbr^d) \;  :  \; w^\T u \geq 0\}$.
As will be utilized in the dual representations presented below,
the Minkowski subtraction for sets $A, B \subseteq M_t$ is defined as
\begin{align*}
A-^{\sbullet} B=\{m \in M_t \; : \; B+\{m\} \subseteq A\}.
\end{align*}

The indicator function for some $D \in \fcal$ is denoted by $\indn{D}: \Omega \rightarrow \{0, 1\}$ and defined as
\begin{align*}
\indn{D}(\omega) =\begin{cases}
1\quad \text{if } \omega \in D\\
0 \quad \text{if } \omega\notin D.
\end{cases}
\end{align*}

\indent Define $\rcal^{\infty, d}$ the space of all $d$-dimensional adapted stochastic processes
$X :=(X_t)_{t \in \bbt}$ on $(\Omega, \fcal, (\fcal_t)_{t \in \bbt}, \P)$ whose coordinates are in $L^\infty(\bbr^d)$ uniformly over $t$.
$\acal^{1, d}$ denotes the space of all $d$-dimensional adapted stochastic processes $a$ on $(\Omega, \fcal, (\fcal_t)_{t \in \bbt}, \P)$ with
$\sum_{i=1}^d \E[\sum_{t\in \bbt} | \Delta
 a_{t, i} |] < +\infty$, where $a_{-1, i}:=0$ and $\Delta a_{t, i}:=a_{t, i} -a_{t-1, i}$.
Denote by $\indn{\bbt_t}$ the element of $\rcal^{\infty, d}$ with the $j$-th component equal to 0 for $j =0, \dots, t-1$ and 1 for $j \in \bbt_t$
and denote by $\ind{s}$ the element of  $\rcal^{\infty, d}$ where only the $s$-th component is equal to 1 and all other are equal to 0.
Similarly, for any $t \in \bbt$, denote $\ind{s\in[0,t)}$ the indicator function of $[0,t)$, i.e., $\ind{s\in[0,t)}$ equals 1 if $s \in [0, t)$ and 0 if $s \notin [0, t)$.
As is standard in the literature (see, e.g.,~\cite{CDK06}), for any time $0 \leq r \leq s \leq T$,
we define the projection $\pi_{r, s}(X)_t:=\indn{\bbt_r}X_{t \wedge s}, t \in \bbt$ for any $d$-dimensional adapted stochastic process $X$.
Denote $\rcal_t^{\infty, d}:=\pi_{t,T}(\rcal^{\infty, d})$  and $\acal^{1, d}_t:= \pi_{t, T}(\acal^{1, d})$.
The space of non-negative processes is written as $\rcal^{\infty, d}_+ :=\{X \in \rcal^{\infty, d} \;  : \; X_{t, i} \geq 0 \text{ for any } t\in \bbt, i=1, \dots, d\}$.
As in~\cite{acciaio2012risk,CDK06}, a process $X \in \rcal_t^{\infty, d}$ describes the evolution of a financial value 
or the cumulative cash flow on the time $\bbt_t$.
Throughout this work we will consider the coarsest topologies
(denoted by $\sigma(\rcal_t^{\infty,d}, \acal_t^{1,d})$,  $\sigma(\acal_t^{1,d}, \rcal_t^{\infty,d})$ respectively) such that $\rcal_t^{\infty,d}$ and $\acal_t^{1,d}$ form a dual pair (and vice versa),
i.e., the space of all continuous linear functionals on the topological space $(\rcal_t^{\infty,d}, \sigma(\rcal_t^{\infty,d}, \acal_t^{1,d}))$ (resp. $(\acal_t^{1,d}$, $\sigma(\acal_t^{1,d}, \rcal_t^{\infty,d}))$ is $\acal_t^{1,d}$ (resp. $\rcal_t^{\infty,d}$).

Denote by $\M(\P)$ the set of probability measures $\Q$ on $(\Omega, \fcal)$
which are absolutely continuous with respect to $\P$.
Denote those probability measures equal to $\P$ on $\fcal_t$ by $\M_t(\P) := \{\Q \in \M(\P) \; : \; \Q = \P \text{ on } \fcal_t\}$.
For any given $t \leq s \in \bbt$, define
\begin{align*}
\D_{t, s} :=\left\{\xi\in L_s^\infty(\bbr_+) \; : \; \E_t[\xi]=1 \right\}.
\end{align*}
Then every $\xi \in \D_{t, s}$ defines a probability measure $\Q^{\xi}$ in $\M_t(\P)$ with density $d \Q^\xi / d \P =\xi$.
Conversely, every $\Q \in \M(\P)$ induces a collection of non-negative random variables $\xi_{t, s}(\Q) \in \D_{t,s}$ for $t \leq s \in \bbt$ where $\xi_{t,s}(\Q)$ is defined as
\begin{align*}
\xi_{t, s}(\Q):=
\begin{cases}\frac{\E_s[\frac{d \Q}{d \P}]}{\E_t[\frac{d \Q}{d \P}]} &\text{on } \{\E_t[\frac{d \Q}{d \P}] > 0\} \\ 1 &\text{otherwise.} \end{cases}
\end{align*}
As in, e.g., Cheridito and Kupper~\cite{CK10}, we will use a $\P$-almost sure version of the $\Q$-conditional expectation of $X \in L^{\infty}(\bbr^d)$ given by
\begin{align*}
\E^{\Q}_t[X]:=\E^{\Q}[X | \fcal_t]:=\E_t[\xi_{t, T}(\Q) X]
\end{align*}
for any time $t \in \bbt$.
For $\Q \in \M(\P)^d$, we will utilize the vector representation $\xi_{t, s}(\Q):=(\xi_{t, s}(\Q_1), \dots, \xi_{t, s}(\Q_d))^\T$ for any time $t \leq s \in \bbt$.
We further define the function $w_t^s: \M(\P)^d \times L_t^1(\bbr^d) \to L_s^0(\bbr^d)$ for any $t, s \in \bbt$ with $t < s$ by
\begin{align*}
w_t^s(\Q, w) &:= \diag(w) \xi_{t, s}(\Q),
\end{align*}
for any $\Q \in \M(\P)^d$ and $w \in L_t^1(\bbr^d)$, where $\diag(w)$ denotes the diagonal matrix with the components of $w$ on the main diagonal.
By convention, for any $t \in \bbt$, $w_t^t(\Q,w)=w$ for any $\Q \in \M(\P)^d$ and  $w \in L_t^1(\bbr^d)$.

As we seek to relate risk measures for random vectors to those of processes, as in~\cite{acciaio2012risk}, we consider the optional filtration, i.e., $(\bar{\Omega},\bar\fcal,(\bar\fcal_t)_{t\in\bbt},\bar\P)$ with sample space $\bar{\Omega}=\Omega \times \bbt$, $\sigma$-algebra $\bar\fcal=\sigma(\{A_t \times \{t\} \; : \;  A_t \in \fcal_t, t \in \bbt\})$,
filtration $\bar\fcal_t=\sigma(\{D_r \times \{r\}, D_t \times \bbt_t \; : \;  D_r \in \fcal_r, r < t, D_t \in \fcal_t \})$, $t \in \bbt$
and probability measure $\bar\P =\P \otimes \mu$ where $\mu=(\mu_t)_{t \in \bbt}$ is some adapted reference process such that $\sum_{t \in \bbt}\mu_t=1$ and $\P[\min_{t \in \bbt} \mu_t > \epsilon] = 1$ for some $\epsilon > 0$.
Under this measure on the optional $\sigma$-algebra, the expectation is
$\bar\E[X]:=\bar\E^{\bar\P}[X]:= \E[\sum_{t \in \bbt} X_t \mu_t]$
for any measurable function $X$ on $(\bar{\Omega}, \bar\fcal)$; we emphasize that this expectation is taken over the optional $\sigma$-algebra by utilizing the $\bar\E$ notation throughout this work.
A random vector $X=(X_t)_{t \in \bbt}$ on $(\bar{\Omega}, \bar\fcal, \bar\P)$ is
$\bar\fcal_t$-measurable if and only if $X_r$ is $\fcal_r$-measurable for all $r=0, \dots, t$
and $X_r=X_t$ for any $r > t$.
Denote $\bar L^p_t(\bbr^d) := L^p(\bar{\Omega},\bar\fcal_t,\bar\P;\bbr^d)$ (with $\bar L^p(\bbr^d) := \bar L^p_T(\bbr^d)$) for $p\in\{0, 1, \infty\}$; similarly denote the set of $\bar\fcal_t$-measurable random vectors taking value in the measurable set $B \subseteq \bbr^d$,
by $\bar L^p_t(B) \subseteq \bar L^p_t(\bbr^d)$ for $p \in\{0,1,\infty\}$.
Often we consider the non-negative orthant $B = \bbr^d_+$ or the positive orthant $B = \bbr^d_{++}$.
Denote $\bar M_t = \bar L_t^\infty(M) := \{X \in \bar L_t^\infty(\bbr^d) \; : \; X \in M \, \bar\P\mbox{-a.s.}\}$.
Throughout this work,  we consider the weak* topology on $\bar L_t^\infty(\bbr^d)$ such that the dual space of $\bar L_t^\infty(\bbr^d)$ is $\bar L_t^1(\bbr^d)$.
As on the original probability space, denote by $\bar\M(\bar\P)$ the set of probability measures $\bar\Q$ on $(\bar{\Omega}, \bar\fcal)$ which are absolutely continuous with respect to $\bar\P$.
Denote by $\bar\E^{\bar\Q}[X]$ as the expectation of $X \in \bar L^\infty(\bbr^d)$ with respect to the probability measure $\bar\Q$.
We will, further, denote $\bar \Gamma_t(\bar w)$ with $\bar w \in \bar L_t^1(\bbr^d) \setminus \{0\}$ by
$\bar \Gamma_t(\bar w) = \{u \in \bar L_t^{\infty}(\bbr^d) \;  :  \; \bar w^\T u \geq 0 \;  \bar\P\mbox{-a.s.}\}$.
The following Theorem \ref{thm:barQ}, adapted from~\cite{acciaio2012risk}, provides a representation for these probability measures on the optional $\sigma$-algebra.

\begin{theorem}
\label{thm:barQ}
For any probability measure $\bar\Q \in \bar\M(\bar\P)$ on the optional $\sigma$-algebra, there exist a probability measure $\Q \in \M(\P)$ and an optional random measure $\psi \in \Psi(\P) := \{\psi \in \rcal_+^{\infty, 1} \; : \; \sum_{t \in \bbt} \psi_t = 1 \; \P\mbox{-a.s.}\}$ such that
\begin{equation}\label{eq:barQ-decomp}
\bar\E^{\bar\Q}[X] = \E^\Q\left[\sum_{t \in \bbt} \psi_t X_t\right]
\end{equation}
for any $X \in \bar L^{\infty}(\bbr)$.

Conversely, any $\Q \in \M(\P)$ and $\psi \in \Psi(\P)$ define a probability measure $\bar\Q := \Q \otimes \psi \in \bar\M(\bar\P)$ such that \eqref{eq:barQ-decomp} holds.
\end{theorem}
\begin{proof}
This follows directly from Theorem 3.4 and Remark B.1 of \cite{acciaio2012risk}.
\end{proof}

\begin{remark}\label{rem:psi}
We wish to note that the decomposition of $\bar\Q \in \bar\M(\bar\P)$ as $\bar\Q = \Q \otimes \psi$ provided in Theorem~\ref{thm:barQ} can be defined with arbitrary $\psi_s \geq 0$ on $\{\E_t[\frac{d\Q}{d\P}] = 0\}$ so long as $\psi \in \Psi(\P)$.  Without loss of generality, we will always consider the modified decomposition $\Q \otimes \hat\psi$ such that:
\begin{align*}
\hat\psi_t = \begin{cases} \psi_t &\text{on } \{\tau(\Q) > t\} \\ \frac{\mu_t}{1 - \sum_{s = 0}^{\tau(\Q)-1} \mu_s}\left(1 - \sum_{s = 0}^{\tau(\Q)-1} \psi_s\right) &\text{on } \{\tau(\Q) \leq t\} \end{cases}
\end{align*}
with stopping time $\tau(\Q) := \min\{t \in \bbt \; | \; \E_t[\frac{d\Q}{d\P}] = 0\}$.
\end{remark}

As we often are interested in those probability measures that are equivalent to $\bar\P$ on the filtration, we present the following corollary.
\begin{corollary}
\label{cor:barQ-t}
Let $\bar\Q \in \bar\M(\bar\P)$ with decomposition $\bar\Q = \Q \otimes \psi$ (as defined in Theorem~\ref{thm:barQ}).  Then, for any time $t \in \bbt$, it follows that
\[\bar\Q = \bar\P \text{ on } \bar\fcal_t \quad \Leftrightarrow \quad \left\{\begin{array}{l}\Q = \P \text{ on } \fcal_t \\ \psi_s = \mu_s \; \forall s < t. \end{array}\right.\]
\end{corollary}
\begin{proof}
This follows directly from \cite[Lemma B.5]{acciaio2012risk}.
\end{proof}

We note that, from Corollary \ref{cor:barQ-t}, it follows that
\begin{align*}
\bar\M_t(\bar\P) &:= \{\bar\Q \in \bar\M(\bar\P) \; : \; \bar\Q = \bar\P \text{ on } \bar\fcal_t\}\\
&= \{\Q \otimes \psi \; : \; \Q \in \M_t(\P), \;\psi \in \Psi(\P), \; \psi_s = \mu_s \; \forall s < t\}.
\end{align*}

As for the original probability space, we will use a $\bar\P$-almost sure version of the $\bar\Q$-conditional expectation of $X \in \bar L^{\infty}(\bbr^d)$.
Denote by $\bar\E^{\bar\Q}_t[X] := \bar\E^{\bar\Q}[X|\bar\fcal_t]$ as the conditional expectation of $X \in \bar L^{\infty}(\bbr^d)$.
This $\bar\P$-almost sure version can be provided by $\bar\xi(\bar\Q)$ construction as for the original probability space, i.e., with
\begin{align}
\label{eq:barxi} \bar\xi_{t,s}(\bar\Q)_r &:= \begin{cases} \frac{\bar\E_s[\frac{d\bar\Q}{d\bar\P}]_r}{\bar\E_t[\frac{d\bar\Q}{d\bar\P}]_r} &\text{on } \{\bar\E_t[\frac{d\bar\Q}{d\bar\P}]_r > 0\} \\ 1 &\text{otherwise} \end{cases}\\
\nonumber &= \begin{cases} \left(\frac{1 - \sum_{\tau = 0}^{t-1} \mu_\tau}{\mu_r}\right)\left(\frac{\psi_r}{1 - \sum_{\tau = 0}^{t-1} \psi_\tau}\right) \xi_{t,r}(\Q) &\text{on } \{r \in [t,s), \; \sum_{\tau = 0}^{t-1} \psi_\tau < 1\} \\ \left(\frac{1 - \sum_{\tau = 0}^{t-1} \mu_\tau}{\sum_{\tau \in \bbt_s} \mu_\tau}\right)\left(\frac{\sum_{\tau \in \bbt_s} \psi_\tau}{1 - \sum_{\tau = 0}^{t-1} \psi_\tau}\right) \xi_{t,s}(\Q) &\text{on } \{r \geq s, \; \sum_{\tau = 0}^{t-1} \psi_\tau < 1\} \\ 1 &\text{otherwise} \end{cases}
\end{align}
under decomposition $\bar\Q = \Q \otimes \psi$ with $\psi$ following the form provided in Remark~\ref{rem:psi}.
In the following Corollary \ref{cor:barQ-conditional} we provide an equivalent representation for the conditional expectation.
\begin{corollary}
\label{cor:barQ-conditional}
For $\bar\Q \in \bar\M(\bar\P)$ with decomposition $\bar\Q = \Q \otimes \psi$, the conditional expectation of $X \in \bar L^\infty(\bbr)$ given $\bar\fcal_t$ takes the form
\begin{equation*}
\bar\E_t^{\bar\Q}[X] = \sum_{s = 0}^{t-1} X_s \ind{s} + \left(\begin{array}{l} \ind{\sum_{r = 0}^{t-1} \psi_r < 1}\E_t^\Q\left[\sum_{s \in \bbt_t} \frac{\psi_s}{1 - \sum_{r = 0}^{t-1} \psi_r} X_s\right]\\ + \ind{\sum_{r = 0}^{t-1} \psi_r = 1}\E_t\left[\sum_{s \in \bbt_t} \frac{\mu_s}{1 - \sum_{r = 0}^{t-1} \mu_r} X_s\right]\end{array}\right)\indn{\bbt_t}.
\end{equation*}
\end{corollary}
\begin{proof}
First, by definition of the conditional expectation and Theorem~\ref{thm:barQ},
\[\E^\Q\left[\sum_{s \in \bbt} \psi_s X_s\right] = \E^\Q\left[\sum_{s = 0}^{t-1} \psi_s \bar\E^{\bar\Q}_t[X]_s + \left(\sum_{s \in \bbt_t} \psi_s\right)\bar\E^{\bar\Q}_t[X]_t\right]\]
for any $X \in \bar L^\infty(\bbr)$.
By matching terms,
\begin{align*}
\bar\E^{\bar\Q}_t[X]_s &= X_s~\Q\mbox{-a.s.} \text{ on } \{\psi_s > 0\}\text{ for } s < t,\\
\bar\E^{\bar\Q}_t[X]_t &= \frac{\E^\Q_t[\sum_{s \in \bbt_t} \psi_s X_s]}{\sum_{s \in \bbt_t} \psi_s}~\Q\mbox{-a.s.} \text{ on } \{\sum_{s \in \bbt_t} \psi_s > 0\}.
\end{align*}
Noting that $\sum_{s \in \bbt_t} \psi_s = 1 - \sum_{s = 0}^{t-1} \psi_s$ by construction, we first recover the $\bar\Q\mbox{-a.s.}$ version of the conditional expectation as presented in, e.g., \cite[Corollary B.3]{acciaio2012risk}.
Finally, we recover the representation of the conditional expectation presented within this corollary by taking the $\bar\P\mbox{-a.s.}$ version of the $\bar\Q$-conditional expectation through the use of $\bar\xi(\bar\Q)$ as presented in \eqref{eq:barxi}; in particular, $\bar\xi_{t,T}(\bar\Q)_r = 1$ for every time $r$ on $\{\sum_{\tau = 0}^{t-1} \psi_\tau = 1\}$. 
\end{proof}

We further define the functions $\bar w_t^s: \bar\M(\bar\P)^d \times \bar L_t^1(\bbr^d) \to \bar L_s^0(\bbr^d)$ for any $t,s \in \bbt$ with $t < s$ by
\begin{align*}
\bar w_t^s(\bar\Q,\bar w) &:= \sum_{r = 1}^{t-1} \bar w_r \ind{r} + \diag(\bar w_t) \left(\sum_{r = t}^{s-1} \bar\xi_{t,s}(\bar\Q)_r \ind{r} + \bar\xi_{t,s}(\bar\Q)_s \indn{\bbt_s}\right)
\end{align*}
where $\bar\xi_{t,s}(\bar\Q)$ is defined in~\eqref{eq:barxi} comparably to $\xi_{t,s}(\Q)$ for $\Q \in \M(\P)^d$.  By convention, for any $t \in \bbt$, $\bar w_t^t(\bar\Q,\bar w) = \bar w$ for any $\bar\Q \in \bar\M(\bar\P)^d$ and $\bar w \in \bar L_t^1(\bbr^d)$.
We note that, through the decomposition of $\bar\Q$ as in Theorem~\ref{thm:barQ}, an equivalent formulation for $\bar w_t^s(\bar\Q,\bar w)$ can be given in the style of Corollary~\ref{cor:barQ-conditional}.

We conclude this section with a brief result relating the topologies of $\bar L^{\infty}(\bbr^d)$ and $\rcal^{\infty,d}$.  This will be utilized later in this work.
\begin{lemma}
\label{lemma:convergence}
\begin{enumerate}
\item Let $X^{n} \to X$ in $\sigma(\bar L^{\infty}(\bbr^d),\bar L^1(\bbr^d))$.  Then $X_t^{n} \to X_t$ in $\sigma(L_t^{\infty}(\bbr^d),L_t^1(\bbr^d))$ and $X^{n}\indn{\bbt_t} \to X \indn{\bbt_t}$ in $\sigma(\rcal_t^{\infty,d},\acal_t^{1,d})$ for every time $t \in \bbt$.
\item Fix $t \in \bbt$.  Let $X_s^{n}  \to X_s$ in $\sigma(L_s^{\infty}(\bbr^d),L_s^1(\bbr^d))$ for every $s < t$ and $X^{n}  \indn{\bbt_t} \to X \indn{\bbt_t}$ in $\sigma(\rcal_t^{\infty,d},\acal_t^{1,d})$.  Then $X^{n}  \to X$ in $\sigma(\bar L^{\infty}(\bbr^d),\bar L^1(\bbr^d))$.
\end{enumerate}
\end{lemma}
\begin{proof}
\begin{enumerate}
\item Let $X^{n}  \to X$ in $\sigma(\bar L^{\infty}(\bbr^d),\bar L^1(\bbr^d))$.  And fix $t \in \bbt$.
    \begin{enumerate}
    \item Let $Z_t \in L_t^1(\bbr^d)$ and define $Y := \mu_t^{-1} Z_t \ind{t} \in \bar L^1(\bbr^d)$. It follows that
        \begin{align*}
        \E[Z_t^\T X_t^{n}] &= \E[\sum_{s \in \bbt} \mu_s Y_s^\T X_s^{n}] = \bar\E[Y^\T X^{n}] \to \bar\E[Y^\T X] = \E[Z_t^\T X_t].
        \end{align*}
    \item Let $a \in \acal_t^{1,d}$.  Define $Y := \sum_{s \in \bbt_t} \mu_s^{-1} \Delta a_s \ind{s} \in \bar L^1(\bbr^d)$. Therefore we find
        \begin{align*}
        \E[\sum_{s \in \bbt_t} \Delta a_s^\T X_s^{n}] &= \E[\sum_{s \in \bbt} \mu_s Y_s^\T X_s^{n}] = \bar\E[Y^\T X^{n}] \to \bar\E[Y^\T X] = \E[\sum_{s \in \bbt_t} \Delta a_s^\T X_s].
        \end{align*}
    \end{enumerate}
\item Fix $t \in \bbt$.  Let $X_s^{n}  \to X_s$ in $\sigma(L_s^{\infty}(\bbr^d),L_s^1(\bbr^d))$ for every $s < t$ and $X^{n}  \indn{\bbt_t} \to X \indn{\bbt_t}$ in $\sigma(\rcal_t^{\infty,d},\acal_t^{1,d})$.  Let $Y \in \bar L^1(\bbr^d)$ and $Z_s := \mu_s Y_s \in L_s^1(\bbr^d)$ for every $s < t$ and $\Delta a_s := \mu_s Y_s$ for every $s \in \bbt_t$ (with $\Delta a_s := 0$ for $s < t$) so that $a_r = \sum_{s = 0}^r \Delta a_r$ defines $a \in \acal_t^{1,d}$. Then we obtain that
    \begin{align*}
    \bar\E[Y^\T X^{n}] &= \E[\sum_{s \in \bbt} \mu_s Y_s^\T X_s^{n}] = \sum_{s = 0}^{t-1} \E[Z_s^\T X_s^{n}] + \E[\sum_{s \in \bbt_t} \Delta a_s^\T X_s^{n}]\\
    &\to \sum_{s = 0}^{t-1} \E[Z_s^\T X_s] + \E[\sum_{s \in \bbt_t} \Delta a_s^\T X_s] = \bar\E[Y^\T X].
    \end{align*}
\end{enumerate}
\end{proof}

\subsection{Risk measures for processes}\label{sec:rm-process}

In this section we provide a quick overview of the definition of set-valued risk measures for processes as defined in~\cite{CH17,CH20}.  Herein we present an axiomatic framework for such functions in Definition~\ref{defn:rm-process}.  We then summarize prior results on the primal representation w.r.t.\ an acceptance set. We conclude this section with considerations for a novel dual representation for these conditional risk measures.

\begin{definition}\label{defn:rm-process}
A function $\rho_t: \rcal_t^{\infty,d} \to \ucal(M_t;M_{t,+})$  for $t \in \bbt$ is  called a \textbf{\emph{set-valued conditional risk measure for processes}} if it satisfies the following properties for all $X, Y  \in  \rcal_t^{\infty,d}$,
\begin{enumerate}
\item Cash invariant: for any $m \in M_t$, 
    \[\rho_t(X + m \indn{\bbt_t})=\rho_t(X)-m;\]
\item Monotone: $\rho_t(X) \subseteq \rho_t(Y) $ if $X \leq Y$ component-wise;
\item Finite at zero: $\rho_t(0) \neq \emptyset$ is closed (w.r.t.\ the subspace topology on $M_t$) and $\fcal_t$-decomposable (i.e., $\indn{A}\rho_t(0) + \indn{A^c}\rho_t(0) \subseteq \rho_t(0)$ for any $A \in \fcal_t$; see, e.g., Chapter 2 of Molchanov~\cite{M05}), and $\P[\tilde\rho_t(0) = M] =0$ where $\tilde\rho_t(0)$ is an $\fcal_t$-measurable random set (i.e., $\operatorname{graph} \tilde\rho_t(0) := \{(\omega, x) \in \Omega \times \bbr^d \; : \; x \in \rho_t(0;\omega)\}$ is $\fcal_t \otimes \mathscr{B}(\bbr^d)$-measurable for Borel $\sigma$-algebra $\mathscr{B}(\bbr^d)$ of $\bbr^d$; see, e.g., Tahar and L\'{e}pinette~\cite{TL12}) such that $\rho_t(0) = L_t^\infty(\tilde\rho_t(0))$.
\end{enumerate}
A conditional risk measure for processes at time $t \in \bbt$ is said to be:
\begin{itemize}
\item Normalized if $\rho_t(X) = \rho_t(X) + \rho_t(0)$ for every $X \in \rcal_t^{\infty,d}$;
\item Conditionally convex if for all $X,Y \in \rcal_t^{\infty,d}$ and $\lambda \in L_t^\infty([0,1])$
\begin{align*}
\rho_t(\lambda X+(1-\lambda)Y) \supseteq \lambda \rho_t(X)+(1-\lambda)\rho_t(Y);
\end{align*}
\item Conditionally positive homogeneous if for all $X \in \rcal_t^{\infty,d}$ and $\lambda \in L_t^\infty(\bbr_{++})$
\begin{align*}
\rho_t(\lambda X)=\lambda\rho_t(X);
\end{align*}
\item Conditionally coherent if it is conditionally convex and conditionally positive homogeneous;
\item Closed if the graph of $\rho_t$
\begin{align*}
\operatorname{graph}\rho_t := \{(X,u) \in \rcal_t^{\infty,d} \times M_t \; : \; u \in \rho_t(X)\}
\end{align*}
is closed in the product topology;
\item Conditionally convex upper continuous (c.u.c.) if
\begin{align*}
\rho_t^-(D) := \{X \in \rcal_t^{\infty,d} \; : \; \rho_t(X) \cap D \neq \emptyset\}
\end{align*}
is closed for any $\fcal_t$-conditionally convex set $D \in \gcal(M_t;-M_{t,+})$.
\end{itemize}
A \textbf{\emph{dynamic risk measure for processes}} is a sequence of conditional risk measures for processes $(\rho_t)_{t \in \bbt}$. And a dynamic risk measure for processes is said to have one of the above properties if $\rho_t$ has the corresponding property for any $t\in\bbt$.
\end{definition}

For the interpretation of these risk measures, recall that $X \in \rcal^{\infty,d}$ denotes the evolution of a financial value or the cumulative cash flow over time.  As such, for instance, cash invariance implies that an addition of $m$ eligible assets at time $t$ to $X \in \rcal^{\infty,d}_t$ which is held until $T$ will transform the value into $X+m\indn{\bbt_t}$ and the risk is decreased by $m$ at time $t$.


We now consider the primal representation of a set-valued conditional risk measure for processes.  That is, for any risk measure, let the \emph{acceptance set} be defined as those claims that do not require the addition of any eligible asset in order to be ``acceptable'' to the risk manager:
\[A_{\rho_t}:=\{X \in \rcal_t^{\infty,d} \; : \; 0 \in \rho_t(X)\}.\]
As shown in~\cite{CH17,CH20}, these acceptance sets uniquely define the risk measure via the relation
\[\rho_t^{A_t}(X) = \{m \in M_t \; : \; X + m\indn{\bbt_t} \in A_t\}\]
for any $X \in \rcal_t^{\infty,d}$.
Furthermore, as provided by~\cite[Proposition 2.1]{CH17}, there is a one-to-one relation between the risk measure and its acceptance set so that
\begin{align*}
\rho_t &= \rho_t^{A_{\rho_t}} \quad \text{ and } \quad A_t = A_{\rho_t^{A_t}}.
\end{align*}
Throughout this work, we will take the risk measure-acceptance set pair $(\rho_t,A_t)$ without use of the explicit sub- or superscript notation.

With this construction of risk measures for processes, we give a consideration of the dual representation.
The following dual representation -- Corollary \ref{cor:dual-process} -- is novel in the literature for set-valued risk measures for processes, we will leave the proof until Section~\ref{sec:equiv-dual}.  Though this result can be proven using the duality theory of Hamel et al.~\cite{setOPsurvey} and following the same logic as the proof of \cite[Corollary 2.4]{FR15-supermtg}, we will take a different approach by utilizing the prior results on risk measures for vectors and the yet-to-be-proven equivalence between these formulations.
For this result, we first need to consider the set of dual variables
\begin{align*}
\W_t &:= \left\{(\Q,w) \in \dcal_t \; : \; w_t^s(\Q_s,w_s) \in L_s^1(\bbr^d_+) \, \forall s \in \bbt_t\right\},\\
\dcal_t &:= \M_t(\P)^{d \times (T-t+1)} \times \left(M_{t,+}^{\ast} \setminus M_t^\perp\right)^{T-t+1}.
\end{align*}
Notably $(\Q,w) \in \W_t$ implies $w_s \geq 0 \; \Q_s\mbox{-a.s.}$ since $w_t^s(\Q_s,w_s) \in L_s^1(\bbr^d_+)$.
\begin{corollary}
\label{cor:dual-process}
A function $\rho_t: \rcal_t^{\infty,d} \to \gcal(M_t;M_{t,+})$ is a \textbf{\emph{closed conditionally convex risk measure}} if and only if
\begin{equation}
\label{eq:dual-process}
\rho_t(X) = \bigcap_{(\Q,w) \in \W_t} \bigg(\sum_{s \in \bbt_t} \Big(\big(\E_t^{\Q_s}[-X_s] + \Gamma_t(w_s)\big) \cap M_t\Big) -^{\sbullet} \alpha_t(\Q,w)\bigg),
\end{equation}
where $\alpha_t$ is the minimal conditional penalty function given by
\begin{equation*}
\alpha_t(\Q,w) = \bigcap_{Y \in A_t} \sum_{s \in \bbt_t} \Big(\big(\E_t^{\Q_s}[-Y_s] + \Gamma_t(w_s)\big) \cap M_t\Big).
\end{equation*}
$\rho_t$ is additionally \textbf{\emph{conditionally coherent}} if and only if
\begin{equation*}
\rho_t(X) = \bigcap_{(\Q,w) \in \W_t^{\max}} \sum_{s \in \bbt_t} \Big(\big(\E_t^{\Q_s}[-X_s] + \Gamma_t(w_s)\big) \cap M_t\Big),
\end{equation*}
for
\begin{equation*}
\W_t^{\max} = \left\{(\Q, w) \in \W_t \; : \;  \sum_{s\in \bbt_t} w_s^\T \E^{\Q_s}_t[Z_s] \geq 0 \ \text{ for any } Z \in A_t\right\}.
\end{equation*}
\end{corollary}

We conclude this section with a brief introduction to a time consistency notion, called multiportfolio time consistency, for these set-valued risk measures for processes.  This notion for risk measures for processes was introduced in~\cite{CH17}.  We will revisit this definition in Section~\ref{sec:mptc}.
\begin{definition}\label{defn:mptc-process}
A dynamic risk measure for processes $(\rho_t)_{t\in\bbt}$ is \textbf{\emph{multiportfolio time consistent}} if for any times $t < s \in \bbt$:
    \[\rho_s(X) \subseteq \bigcup_{Y \in B} \rho_s(Y) \; \Rightarrow \; \rho_t(Z\ind{[t,s)} + X\indn{\bbt_s}) \subseteq \bigcup_{Y \in B} \rho_t(Z\ind{[t,s)} + Y\indn{\bbt_s})\]
    for any $X \in \rcal_s^{\infty,d}$, $Z \in \rcal_t^{\infty,d}$, and $B \subseteq \rcal_s^{\infty,d}$.
\end{definition}

Conceptually, a risk measure is multiportfolio time consistent if a cumulative cash flow process $X$ is guaranteed (almost surely) to be more risky than some collection of cumulative cash flows $B$ at a time point in the future $s$ and if all of those cumulative cash flows are identical up to time $s$, then $X$ is also more risky at all earlier times.  This can most easily be understood when the collection $B$ is a singleton, i.e., when comparing two cumulative cash flows $X$ and $Y$ directly.

\subsection{Risk measures for vectors}\label{sec:rm-vector}
In this section we provide a quick overview of the definition of set-valued risk measures for random vectors as defined in~\cite{FR12,FR12b}.  Herein we present an axiomatic framework for such functions in Definition~\ref{defn:rm-vector} with emphasis on the special case on the optional filtration.  We then summarize prior results on the primal representation w.r.t.\ an acceptance set. We conclude this section with considerations for a dual representation for these conditional risk measures.

\begin{definition}\label{defn:rm-vector}
A function $\bar R_t: \bar L^{\infty}(\bbr^d) \to \ucal(\bar M_t;\bar M_{t,+})$  for $t \in \bbt$ is  called a \textbf{\emph{set-valued conditional risk measure for vectors}} on $\bar L^{\infty}(\bbr^d)$ if it satisfies the following properties for all $X, Y  \in \bar L^{\infty}(\bbr^d)$,
\begin{enumerate}
\item Cash invariant: for any $m \in \bar M_t$, 
    \[\bar R_t(X + m)=\bar R_t(X)-m;\]
\item  Monotone: $\bar R_t(X) \subseteq \bar R_t(Y)$ if $X \leq Y$ component-wise;
\item  Finite at zero: $\bar R_t(0) \neq \emptyset$ is closed and $\bar\fcal_t$-decomposable, and such that $\bar\P[\tilde{\bar R}_t(0) = M] = 0$ where $\tilde{\bar R}_t(0)$ is a $\bar\fcal_t$-measurable random set such that $\bar R_t(0) = \bar L_t^\infty(\tilde{\bar R}_t(0))$.
\end{enumerate}
A conditional risk measure for vectors at time $t \in \bbt$ is said to be:
\begin{itemize}
\item Normalized if $\bar R_t(X) = \bar R_t(X) + \bar R_t(0)$ for every $X \in \bar L^{\infty}(\bbr^d)$;
\item Time decomposable if
\begin{align*}
\bar R_t(X) = \sum_{s = 0}^{t-1} \bar R_t(X \ind{s})\ind{s} + \bar R_t(X \indn{\bbt_t})\indn{\bbt_t}
\end{align*}
for any $X \in \bar L^{\infty}(\bbr^d)$;
\item Conditionally convex if for all $X,Y \in \bar L^{\infty}(\bbr^d)$ and $\lambda \in \bar L_t^\infty([0,1])$
\begin{align*}
\bar R_t(\lambda X+(1-\lambda)Y) \supseteq \lambda \bar R_t(X)+(1-\lambda)\bar R_t(Y);
\end{align*}
\item Conditionally positive homogeneous if for all $X \in \bar L^{\infty}(\bbr^d)$ and $\lambda \in \bar L_t^\infty(\bbr_{++})$
\begin{align*}
\bar R_t(\lambda X)=\lambda\bar R_t(X);
\end{align*}
\item Conditionally coherent if it is conditionally convex and conditionally positive homogeneous;
\item Closed if the graph of $\bar R_t$
\begin{align*}
\operatorname{graph}\bar R_t := \{(X,u) \in \bar L^{\infty}(\bbr^d) \times \bar M_t \; : \; u \in \bar R_t(X)\}
\end{align*}
is closed in the product topology;
\item Conditionally convex upper continuous (c.u.c.) if
\begin{align*}
\bar R_t^-(D) := \{X \in \bar L^{\infty}(\bbr^d) \; : \; \bar R_t(X) \cap D \neq \emptyset\}
\end{align*}
is closed for any $\bar\fcal_t$-conditionally convex set $D \in \gcal(\bar M_t;-\bar M_{t,+})$.
\end{itemize}
A \textbf{\emph{dynamic risk measure for vectors}} is a sequence of conditional risk measures for vectors $(\bar R_t)_{t \in \bbt}$. A dynamic risk measure for vectors is said to have one of the above properties if $\bar R_t$ has the corresponding property for any $t\in\bbt$.
\end{definition}

\begin{remark}\label{rem:td}
We note that the time decomposable property presented above is, as far as the authors are aware, novel to this work. Time decomposability conceptually means that the time $s < t \in \bbt$ eligible assets required to compensate for the risk of the cumulative cash flows at time $s$ only depend on the realized portfolio at time $s$; this similarly is true for the eligible assets required to compensate for the random future cash flows at time $t$.
 Notably, by~\cite[Proposition 2.8]{FR12}, any conditionally convex risk measure is time decomposable.
%
\end{remark}

We now consider the primal representation of a set-valued conditional risk measure for vectors.  That is, for any risk measure, let the \emph{acceptance set} be defined as those claims that do not require the addition of any eligible asset in order to be ``acceptable'' to the risk manager:
\[\bar A_{\bar R_t}:=\{X \in \bar L^{\infty}(\bbr^d) \; : \; 0 \in \bar R_t(X)\}.\]
As shown in~\cite{FR12}, these acceptance sets uniquely define the risk measure via the relation
\[\bar R_t^{\bar A_t}(X) = \{m \in \bar M_t \; : \; X + m \in \bar A_t\}\]
for any $X \in \bar L^{\infty}(\bbr^d)$.
Furthermore, as with the risk measures for processes and as provided by~\cite[Remark 2]{FR12}, there is a one-to-one relation between the risk measure and its acceptance set so that
\begin{align*}
\bar R_t &= \bar R_t^{\bar A_{\bar R_t}} \quad \text{ and } \quad \bar A_t = \bar A_{\bar R_t^{\bar A_t}}.
\end{align*}
Throughout this work, we will take the risk measure-acceptance set pair $(\bar R_t,\bar A_t)$ without use of the explicit sub- or superscript notation.

We now provide a dual representation for risk measures for vectors.
For comparison to the dual representation presented above for risk measures for processes, we present the dual representation provided first in~\cite{FR15-supermtg}.
For this result, we first need to consider the set of dual variables
\[\bar\W_t := \left\{(\bar\Q,\bar w) \in \bar\M_t(\bar\P)^d \times (\bar M_{t,+}^{\ast}\setminus \bar M_t^\perp) \; : \; \bar w_t^T(\bar\Q,\bar w) \in \bar L^1(\bbr^d_+)\right\}.\]
As above, $(\bar\Q,\bar w) \in \bar\W_t$ implies $\bar w \geq 0 \; \bar\Q\mbox{-a.s.}$
\begin{corollary}\cite[Corollary 2.4]{FR15-supermtg}
\label{cor:dual-vector}
A function $\bar R_t: \bar L^{\infty}(\bbr^d) \to \gcal(\bar M_t;\bar M_{t,+})$ is a \textbf{\emph{closed conditionally convex risk measure}} if and only if
\begin{equation*}
\bar R_t(X) = \bigcap_{(\bar\Q,\bar w) \in \bar\W_t} \Big(\big(\bar\E_t^{\bar\Q}[-X] + \bar\Gamma_t(\bar w)\big) \cap \bar M_t -^{\sbullet} \bar\alpha_t(\bar\Q,\bar w)\Big),
\end{equation*}
where $\bar\alpha_t$ is the minimal conditional penalty function given by
\begin{equation*}
\bar\alpha_t(\bar\Q,\bar w) = \bigcap_{Y \in \bar A_t} \left(\bar\E_t^{\bar\Q}[-Y] + \bar\Gamma_t(\bar w)\right) \cap \bar M_t.
\end{equation*}
$\bar R_t$ is additionally \textbf{\emph{conditionally coherent}} if and only if
\begin{equation*}
\bar R_t(X) = \bigcap_{(\bar\Q,\bar w) \in \bar\W_t^{\max}} \left(\bar\E_t^{\bar\Q}[-X] + \bar\Gamma_t(\bar w)\right) \cap \bar M_t,
\end{equation*}
for
\begin{equation*}
\bar\W_t^{\max} = \left\{(\bar\Q,\bar w) \in \bar\W_t \; : \; \bar w_t^T(\bar\Q,\bar w) \in \bar A_t^{\ast}\right\}.
\end{equation*}
\end{corollary}

We conclude this introduction to risk measures for vectors with a consideration of multiportfolio time consistency.
This notion for risk measures for vectors is a modification of that introduced in~\cite{FR12,FR12b}; we demonstrate in Proposition~\ref{prop:mptc-vector} that this coincides with the typical definition for time decomposable risk measures.
\begin{definition}\label{defn:mptc-vector}
A dynamic risk measure for vectors $(\bar R_t)_{t\in\bbt}$ is \textbf{\emph{multiportfolio time consistent}} if for any times $t < s \in \bbt$:
    \[\bar R_s(X) \subseteq \bigcup_{Y \in B} \bar R_s(Y) \; \Rightarrow \; \bar R_t(X) \subseteq \bigcup_{Y \in B} \bar R_t(Y)\]
    for any $X \in \bar L^\infty$ and $B := \{(b_0,b_1,\dots,b_{s-1},b_s) \; : \; b_r \in B_r, r < s, \; b_s \in B_s\}$ for any arbitrary sequence of sets $B_r \subseteq L_r^\infty(\bbr^d)$ for $r < s$ and $B_s \subseteq \rcal^{\infty,d}_s$.
\end{definition}

As for risk measures for processes, conceptually, multiportfolio time consistency allows for the comparison of a single (cumulative) cash flow with a collection of portfolios over time in which a guaranteed ordering must hold backwards in time.

\begin{proposition}\label{prop:mptc-vector}
Let $(\bar R_t)_{t\in\bbt}$ be a normalized, time decomposable risk measure.  $(\bar R_t)_{t\in\bbt}$ is multiportfolio time consistent if and only if for any times $t < s \in \bbt$:
    \[\bar R_s(X) \subseteq \bigcup_{Y \in B} \bar R_s(Y) \; \Rightarrow \; \bar R_t(X) \subseteq \bigcup_{Y \in B} \bar R_t(Y)\]
    for any $X \in \bar L^{\infty}(\bbr^d)$ and $B \subseteq \bar L^{\infty}(\bbr^d)$.
\end{proposition}
\begin{proof}
This result follows trivially via the application of the recursive relation $\bar R_t(X) = \bigcup_{Z \in \bar R_s(X)} \bar R_t(-Z)$ for any times $t < s$ and $X \in \bar L^{\infty}(\bbr^d)$ which can be shown to be equivalent to multiportfolio time consistency as given in Definition~\ref{defn:mptc-vector} and the more general version provided in the statement of this proposition. 
\end{proof}

\begin{remark}\label{rem:Rt}
Throughout much of the remainder of this work, the restriction of set-valued risk measures for vectors to a sub-$\sigma$-algebra is utilized.  Herein we will define $R_t: L_t^{\infty}(\bbr^d) \to \ucal(M_t;M_{t,+})$ to be the restriction of a set-valued risk measure (with filtered probability space $(\Omega,\fcal,(\fcal_s)_{s \in \bbt},\P)$) to $L_t^{\infty}(\bbr^d)$.
The primal and dual representation of such risk measures will be utilized extensively below defined in analogous ways to those provided above; for closed and convex risk measures -- as are utilized for the dual representation -- the codomain of $R_t$ is $\gcal(M_t;M_{t,+}) \subseteq \ucal(M_t;M_{t,+})$.
Notably, the penalty function $\alpha_{R_t}$ for these restricted risk measures only depends on a single dual variable $w \in M_{t,+}^{\ast}$, i.e., $\alpha_{R_t}(w)$ does \emph{not} depend on a vector of probability measures $\Q \in \M_t(\P)^d$.
\end{remark}

\section{Equivalence of risk measures for processes and vectors}\label{sec:equiv}

In this section we will focus on the relation between set-valued risk measures for processes and vectors.  This was studied under restrictive conditions in \cite{CH17} which match many of the same properties of the scalar setting (i.e., with a full space of eligible assets $M = \bbr^d$ and with the strong normalization property $\rho_t(0) = L_t^{\infty}(\bbr^d_+)$).
Herein we drop these restrictions and determine the equivalence for both the primal and dual representations.  Notably, this equivalence requires the introduction of an augmentation of the risk measure for processes with a series of risk measures for vectors on the original filtered probability space; when the aforementioned strong assumptions are imposed these augmenting risk measures can be dropped without incident as detailed in Corollary~\ref{cor:equiv-primal-full}.

\subsection{Relation for primal representations}\label{sec:equiv-primal}
First, we compare the relation between set-valued risk measures for processes and for vectors through their axiomatic definitions.  This is akin to the results presented in~\cite{acciaio2012risk} for scalar risk measures.  Notably, in comparison to that work, the set-valued risk measures for vectors in the optional filtration are a larger class than those for processes rather than the equivalence as found for scalar risk measures.  In particular, we find risk measures for vectors on the optional filtration are equivalent to the risk measures for processes augmented with the restricted set-valued risk measures for vectors which are discussed in Remark~\ref{rem:Rt}.  Intuitively, this augmentation is necessary to account for the possible risk prior to the current time that cannot trivially be accounted for by eligible assets; that is, the augmentation is required for the $M \neq \bbr^d$ setting.  In the full eligible asset setting $M = \bbr^d$, the equivalence of risk measures for processes and vectors no longer requires an augmentation by restricted set-valued risk measures as is detailed in Corollary~\ref{cor:equiv-primal-full} below.
\begin{theorem}\label{thm:equiv-primal}
Fix time $t \in \bbt$.
\begin{enumerate}
\item\label{thm:equiv-primal-1} Any conditional risk measure for processes $\rho_t: \rcal_t^{\infty,d} \to \ucal(M_t;M_{t,+})$ and series of conditional risk measures for vectors $R_s: L_s^{\infty}(\bbr^d)  \to \ucal(M_s;M_{s,+})$ (restricted to $\fcal_s$-measurable random vectors), $s=0, 1, \dots, t-1$, define a time decomposable conditional risk measure for random vectors on the optional filtration $\bar R_t: \bar L^{\infty}(\bbr^d) \to \ucal(\bar M_t;\bar M_{t,+})$ via
\begin{align}
\label{eq:primal-barR} \bar R_t(X) &:= \sum_{s = 0}^{t-1} R_s(X_s)\ind{s} + \rho_t(\pi_{t,T}(X))\indn{\bbt_t}.
\end{align}
Additionally, if $(R_s)_{s = 0}^{t-1},\rho_t$ are normalized, conditionally convex, or conditionally coherent, then the corresponding risk measure $\bar R_t$ has the same property.

\item\label{thm:equiv-primal-2} Any conditional risk measure for random vectors on the optional filtration $\bar R_t: \bar L^{\infty}(\bbr^d) \to \ucal(\bar M_t;\bar M_{t,+})$ defines a conditional risk measure for processes $\rho_t: \rcal_t^{\infty,d} \to \ucal(M_t;M_{t,+})$ and a series of conditional risk measures $R_s: L_s^{\infty}(\bbr^d) \to \ucal(M_s;M_{s,+})$ restricted to $\fcal_s$-measurable random vectors, $s = 0,1,\dots,t-1$,  via
\begin{align}
\label{eq:primal-rho} \rho_t(X) &:= \bar R_t(X\indn{\bbt_t})_t\\
\label{eq:primal-R} R_s(Z) &:= \bar R_t(Z \ind{s})_s, \quad s = 0,1,\dots,t-1.
\end{align}
Additionally, if $\bar R_t$ is normalized, conditionally convex, or conditionally coherent, then the corresponding risk measures $(R_s)_{s = 0}^{t-1},\rho_t$ have the same property.

\item\label{thm:equiv-primal-3} If $\bar R_t$ is time decomposable then it can be decomposed into form \eqref{eq:primal-barR} such that $\rho_t$ and $(R_s)_{s = 0}^{t-1}$ are defined as in \eqref{eq:primal-rho} and \eqref{eq:primal-R} respectively.  Additionally, if $\bar R_t$ is conditionally convex and either closed or conditionally convex upper continuous then $\rho_t$ and $(R_s)_{s = 0}^{t-1}$ have the same property and vice versa.
\end{enumerate}
\end{theorem}
\begin{proof}
\begin{enumerate}
\item It is easy to check that $\bar R_t$ defined in \eqref{eq:primal-barR} is a time decomposable conditional risk measure on the optional filtration as defined in Definition~\ref{defn:rm-vector}.  This is done directly by considering the equivalent definitions for the risk measure for processes and the sequence of conditional risk measures. Due to simplicity, we omit the direct constructions herein. 

\item It is easy to check that $\rho_t$ defined in \eqref{eq:primal-rho} is a risk measure for processes and $R_s$, for $s < t$, defined in \eqref{eq:primal-R} is a conditional risk measure for vectors on $L_s^{\infty}(\bbr^d)$.  Notably, for $\rho_t$, the result follows directly from \cite[Proposition 3.4]{CH17} with only verification required as we no longer require the full eligible space, i.e., $M \subseteq \bbr^d$; due to simplicity of those results, we omit those proofs.  However, due to the updated definition for finiteness at zero, that result needs to be proven for $\rho_t$ separately from the approach of \cite{CH17} but still follows trivially by construction (recalling that $\P[\min_{t \in \bbt} \mu_t > \epsilon] = 1$ for some $\epsilon > 0$) and thus we omit the details.

\item If $\bar R_t$ is time decomposable and $X \in \bar L^{\infty}(\bbr^d)$ then:
    \begin{align*}
    \bar R_t(X) &= \sum_{s = 0}^{t-1} \bar R_t(X\ind{s})\ind{s} + \bar R_t(X\indn{\bbt_t})\indn{\bbt_t}\\
    &= \sum_{s = 0}^{t-1} R_s(X_s)\ind{s} + \rho_t(\pi_{t,T}(X))\indn{\bbt_t}.
    \end{align*}

Finally, we want to study the closedness properties for conditionally convex risk measures.
    \begin{enumerate}
    \item \textbf{Closed}:
        \begin{itemize}
        \item Assume $\bar R_t$ is closed.
            First, we will show that $R_s$ is closed for any $s < t$.
            \begin{align*}
            \operatorname{graph} R_s &= \{(X_s,m_s) \in L_s^{\infty}(\bbr^d) \times M_s \; : \; m_s \in R_s(X_s)\}\\
                &= \{(X_s,m_s) \in L_s^{\infty}(\bbr^d) \times M_s \; : \; m_s \in \bar R_t(X_s \ind{s})_s\}\\
                &= \{(X,m) \in \bar L^{\infty}(\bbr^d) \times \bar M_t \; : \; m \in \bar R_t(X)\}_s\\
                &= (\operatorname{graph} \bar R_t)_s.
            \end{align*}
            Let $(X_s^{n},m_s^{n})_{n \in N} \subseteq \operatorname{graph} R_s \to (X_s,m_s)$ in $\sigma(L_s^{\infty}(\bbr^d),L_s^1(\bbr^d))$ be a convergent net indexed by the set $N$.  By Lemma~\ref{lemma:convergence} and assumption, it must follow that $(X_s^{n}\ind{s},m_s^{n}\ind{s} + m^0\indn{\{s\}^c}) \to (X_s \ind{s},m_s \ind{s} + m^0\indn{\{s\}^c}) \in \operatorname{graph} \bar R_t$ for any arbitrary $m^0 \in \bar R_t(0)$.  As $(X_s,m_s) \in \operatorname{graph} R_s$, closedness is proven.
            %
            Consider now the closedness of $\rho_t$.  By the same logic as above,
            \begin{align*}
            \operatorname{graph} \rho_t &= \{(X\indn{\bbt_t},m_t) \; : \; (X,m) \in \operatorname{graph} \bar R_t\}.
            \end{align*}
            Therefore, as above by concatenation with $(0,m^0) \in \operatorname{graph} \bar R_t$, $\operatorname{graph} \rho_t$ is closed.
        \item Assume $\rho_t$ and $(R_s)_{s = 0}^{t-1}$ are all closed risk measures.
            \begin{align*}
            \operatorname{graph} \bar R_t &= \{(X,m) \in \bar L^{\infty}(\bbr^d) \times \bar M_t \; : \; m \in \bar R_t(X)\}\\
                &= \left\{(X,m) \in \bar L^{\infty}(\bbr^d) \times \bar M_t \; : \; \begin{array}{l} m_s \in R_s(X_s) \; \forall s < t, \\ m_t \in \rho_t(\pi_{t,T}(X))\end{array}\right\}\\
                &= \sum_{s = 0}^{t-1} (\operatorname{graph} R_s) \ind{s} + (\operatorname{graph} \rho_t) \indn{\bbt_t}.
            \end{align*}
            Let $(X^{n},m^{n}) \subseteq \operatorname{graph} \bar R_t \to (X,m)$ in the product topology.  By the above construction and Lemma~\ref{lemma:convergence}, it must follow that $(X_s,m_s) \in \operatorname{graph} R_s$ for every time $s < t$ and $(X\indn{\bbt_t},m_t) \in \operatorname{graph} \rho_t$.  Therefore closedness is proven as $(X,m) \in \operatorname{graph} \bar R_t$.
        \end{itemize}
    \item \textbf{Conditionally convex upper continuous}:
        \begin{itemize}
        \item Assume $\bar R_t$ is conditionally convex upper continuous.
            First we will show that $R_s$ is conditionally convex upper continuous.  Let $D_s \in \gcal(M_s;-M_{s,+})$ be closed and conditionally convex.
            \begin{align*}
            R_s^-(D_s) &= \{X_s \in L_s^{\infty}(\bbr^d) \; : \; R_s(X_s) \cap D_s \neq \emptyset\}\\
                &= \{X_s \in L_s^{\infty}(\bbr^d) \; : \; \bar R_t(X_s \ind{s})_s \cap D_s \neq \emptyset\}\\
                &= \{X \in \bar L^{\infty}(\bbr^d) \; : \; \bar R_t(X)_s \cap D_s \neq \emptyset\}_s\\
                &= \bar R_t^-(D_s \ind{s} + \bar M_t \indn{\{s\}^c})_s.
            \end{align*}
            Let $X_s^{n} \subseteq R_s^-(D_s) \to X_s$ in $\sigma(L_s^{\infty}(\bbr^d),L_s^1(\bbr^d))$.  By Lemma~\ref{lemma:convergence} and assumption, it must follow that $X_s^{n} \ind{s} \to X_s \ind{s} \in \bar R_t^-(D_s \ind{s} + \bar M_t \indn{\{s\}^c})$.  Therefore conditional convex upper continuity is proven as $X_s \in R_s^-(D_s)$.
            Consider now the conditional convex upper continuity of $\rho_t$.  Let $D_t \in \gcal(M_t;-M_{t,+})$ be closed and conditionally convex.  By the same logic as above,
            \begin{align*}
            \rho_t^-(D_t) &= \bar R_t^-(\sum_{s = 0}^{t-1} M_s \ind{s} + D_t \indn{\bbt_t}) \indn{\bbt_t}.
            \end{align*}
            Therefore, as above, $\rho_t^-(D_t)$ must be closed and the result follows.
        \item Assume $\rho_t$ and $(R_s)_{s = 0}^{t-1}$ are all conditionally convex upper continuous.
            Let $D \in \gcal(\bar M_t;-\bar M_{t,+})$ be closed and conditionally convex.  By decomposability in the optional filtration, $D = \sum_{s = 0}^{t-1} D_s \ind{s} + D_t \indn{\bbt_t}$ for $D_s \in \gcal(M_s;-M_{s,+})$ closed and conditionally convex for every time $s \leq t$. Hence,
            \begin{align*}
            \bar R_t^-(D) &= \{X \in \bar L^{\infty}(\bbr^d) \; : \; \bar R_t(X) \cap D \neq \emptyset\}\\
                &= \{X \in \bar L^{\infty}(\bbr^d) \; : \; \bar R_t(X)_s \cap D_s \neq \emptyset \; \forall s \leq t\}\\
                &= \{X \in \bar L^{\infty}(\bbr^d) \; : \; R_s(X_s) \cap D_s \neq \emptyset \; \forall s < t, \; \rho_t(\pi_{t,T}(X)) \cap D_t \neq \emptyset\}\\
                &= \sum_{s = 0}^{t-1} R_s^-(D_s) \ind{s} + \rho_t^-(D_t) \indn{\bbt_t}.
            \end{align*}
            Let $X^{n} \subseteq \bar R_t^-(D) \to X$ in $\sigma(\bar L^{\infty}(\bbr^d),\bar L^1(\bbr^d))$.  By Lemma~\ref{lemma:convergence} and assumption, $X_s^{n} \to X_s \in R_s^-(D_s)$ for every $s < t$ and $X^{n} \indn{\bbt_t} \to X \indn{\bbt_t} \in \rho_t^-(D_t)$.  Therefore conditional convex upper continuity is proven as $X \in \bar R_t^-(D)$.
        \end{itemize}
    \end{enumerate}
\end{enumerate}
\end{proof}

Given these results on the equivalence of risk measures for processes and vectors, we consider the special case in which every asset is eligible to cover risk, i.e., $M = \bbr^d$.  This setting provides a simpler equivalent relation insofar as the restricted risk measures $R_s$ utilized in Theorem~\ref{thm:equiv-primal} above can be fully defined by its value at $0$.  This setting more clearly demonstrates the relationship with scalar risk measures, as $R_s(0) = L_s^{\infty}(\bbr^d_+)$ in the $d = 1$ asset setting for any choice of normalized risk measure with closed values.
\begin{corollary}\label{cor:equiv-primal-full}
Let $M = \bbr^d$.  The equivalences in Theorem~\ref{thm:equiv-primal} can be simplified insofar as the conditional risk measures $R_s$ can be replaced with (conditionally convex) upper sets $C_s$, i.e.,
    \begin{align*}
    C_s &:= \bar R_t(0)_s\\
    \bar R_t(X) &= \sum_{s = 0}^{t-1} (-X_s + C_s)\ind{s} + \rho_t(\pi_{t,T}(X))\indn{\bbt_t}.
    \end{align*}
\end{corollary}
\begin{proof}
This follows from Theorem~\ref{thm:equiv-primal} by noting that $R_s(X \ind{s}) = R_s(0) - X_s$ by cash invariance for this choice of eligible assets.  Therefore the full risk measure $R_s$ is no longer required, but only the risk measure at $0$ which we denote $C_s$. 
\end{proof}

We conclude this section on the equivalence of primal representations by providing the relation between acceptance sets for processes and vectors on the optional filtration.
\begin{corollary}\label{cor:equiv-primal-acceptance}
Fix time $t \in \bbt$.
\begin{enumerate}
\item Consider a conditional risk measure for processes $\rho_t: \rcal_t^{\infty,d} \to \ucal(M_t;M_{t,+})$ and series of conditional risk measures $R_s: L_s^{\infty}(\bbr^d) \to \ucal(M_s;M_{s,+})$ for $s=0, 1, \dots, t-1$, and $\bar R_t$ is defined as \eqref{eq:primal-barR}, then
$\bar A_t = \sum_{s = 0}^{t-1} A_{R_s} \ind{s} + A_t \indn{\bbt_t}$.
\item  Given a conditional risk measure for random vectors on the optional filtration $\bar R_t: \bar L^{\infty}(\bbr^d) \to \ucal(\bar M_t;\bar M_{t,+})$, define $\rho_t$ and $R_s$ via \eqref{eq:primal-rho} and \eqref{eq:primal-R} respectively, then $A_t = \{X \in \rcal_t^{\infty,d} \; : \; X\indn{\bbt_t} \in \bar A_t\}$ and $A_{R_s} = (\bar A_t)_s$ for every $s = 0,1,\dots,t-1$.
\end{enumerate}
\end{corollary}
\begin{proof}
This result follows directly by construction of the equivalences in Theorem~\ref{thm:equiv-primal}. 
\end{proof}

\subsection{Relation for dual representations}\label{sec:equiv-dual}

Before presenting the following results, we want the reader to recall the dual representations for risk measures for processes in Corollary~\ref{cor:dual-process} and vectors in Corollary~\ref{cor:dual-vector}.  We also highlight the form for the penalty functions for risk measures $R_s: L_s^{\infty}(\bbr^d) \to \gcal(M_s;M_{s,+})$ as provided in Remark~\ref{rem:Rt}.  Finally, we remind the reader that any conditionally convex risk measure for vectors is time decomposable by construction as given in Remark~\ref{rem:td}.

\begin{lemma}\label{lemma:equiv-dual}
Fix time $t \in \bbt$.
\begin{enumerate}
\item\label{lemma:equiv-dual-1} Let $\rho_t: \rcal_t^{\infty,d} \to \gcal(M_t;M_{t,+})$ be a closed conditionally convex risk measure for processes and let $R_s: L_s^{\infty}(\bbr^d) \to \gcal(M_s;M_{s,+})$, $s < t$, be a series of closed conditionally convex risk measures. The associated closed conditionally convex risk measure on the optional filtration $\bar R_t: \bar L^{\infty}(\bbr^d) \to \gcal(\bar M_t;\bar M_{t,+})$ defined in~\eqref{eq:primal-barR} has dual representation w.r.t.\ the penalty function
    \begin{align*}
    \bar\alpha_t(\bar\Q, \bar w) = \bar\alpha_t(\Q\otimes\psi,\bar w) = \sum_{s=0}^{t-1} \alpha_{R_s}(\bar w_s) \ind{s} + \alpha_t(W_t(\Q\otimes\psi,\bar w))\indn{\bbt_t}
    \end{align*}
    where $W_t(\Q\otimes\psi,\bar w) := (\hat\Q, w) \in \W_t$ is defined as
    \begin{align*}
    \frac{d\hat\Q_{s,i}}{d\P} &:= \frac{\psi_{s,i}}{\E_t^{\Q_i}[\psi_{s,i}]}\xi_{t,s}(\Q_i) \ind{\E_t^{\Q_i}[\psi_{s,i}] > 0} + \frac{\mu_s}{\E_t[\mu_s]} \ind{\E_t^{\Q_i}[\psi_{s,i}] = 0}, \\
    w_{s,i} &:= \frac{\E_t^{\Q_i}[\psi_{s,i}]}{1-\sum_{r = 0}^{t-1} \mu_r} \bar{w}_{t,i}
    \end{align*}
    for indices $i = 1,2,\dots,d$ and times $s \in \bbt_t$.

\item\label{lemma:equiv-dual-2} Consider a closed conditionally convex risk measure for random vectors on the optional filtration $\bar R_t: \bar L^{\infty}(\bbr^d) \to \gcal(\bar M_t;\bar M_{t,+})$.  The associated closed conditionally convex risk measure for processes $\rho_t: \rcal_t^{\infty,d} \to \gcal(M_t;M_{t,+})$ and series of closed conditionally convex risk measures
    $R_s: L_s^{\infty}(\bbr^d) \to \gcal(M_s;M_{s,+})$ defined in~\eqref{eq:primal-rho} and~\eqref{eq:primal-R}, respectively, have dual representations w.r.t.\ the penalty functions
    \begin{align*}
    \alpha_{R_s}(w_s) &=\bar\alpha_t(\bar\P,w_s\ind{s})_s,\\
    \alpha_t(\Q, w) &=\bar{\alpha}_t(\bar W_t(\Q,w))_t,
    \end{align*}
    where $\bar W_t(\Q,w) := (\bar\Q, \bar w) \in \bar\W_t$ is defined as
    \begin{align*}
    \bar w &:= \sum_{s\in\bbt_t} \diag\left(\ind{w_s > 0}\right) w_s \indn{\bbt_t},\\
    \left(\frac{d\bar\Q_i}{d\bar\P}\right)_s &:= \ind{s \in [0,t)} + \frac{1 - \sum_{r = 0}^{t-1} \mu_r}{\mu_s} \left(\frac{w_{s,i}}{\bar w_{t,i}} \ind{\bar w_{t,i} > 0} + \ind{\bar w_{t,i} = 0}\right) \xi_{t,s}(\Q_{s,i}) \indn{\bbt_t}
    \end{align*}
    for indices $i = 1,2,\dots,d$.
\end{enumerate}
\end{lemma}
\begin{proof}
\begin{enumerate}
\item We will prove this result by first demonstrating the decomposition of the penalty function $\bar\alpha_t$; once that is demonstrated we will then show that $W_t(\Q \otimes \psi,\bar w) \in \W_t$ guaranteeing that the decomposition is well-defined.
Recall the definition of the minimal penalty function $\bar\alpha_t$ from Corollary~\ref{cor:dual-vector}, i.e.,
\begin{align*}
\bar{\alpha}_t(\bar\Q, \bar w)&=\bigcap_{Z \in \bar A_t}  \left(\bar\E_t^{\bar\Q}[-Z] + \bar\Gamma_t(\bar w)\right) \cap \bar M_t\\
&=\left\{u \in \bar M_t \; : \; \bar w^\T  u \geq \esssup_{Z \in \bar A_t} \bar w^\T  \bar\E_t^{\bar\Q}[-Z] \; \bar\P\mbox{-a.s.} \right\}.
\end{align*}
Consider now the decomposition of the acceptance set $\bar A_t = \sum_{s = 0}^{t-1} A_{R_s}\ind{s} + A_t \indn{\bbt_t}$ as given in Corollary~\ref{cor:equiv-primal-acceptance}.  Additionally, assume $\bar\Q_i = \Q_i \otimes \psi_i$ be the decomposition as provided in Corollary~\ref{cor:barQ-t} for each component $i$ of the vector of measures $\bar\Q \in \bar\M_t(\bar\P)^d$. Recall from Corollary~\ref{cor:barQ-t} that $\psi_{s,i} = \mu_s$ for every asset $i$ and time $s < t$.  Additionally, and as a direct consequence of this equality, $\sum_{s = 0}^{t-1} \psi_{s,i} < 1$ for every asset $i$. From this relation, it follows that
\begin{align*}
&\bar\alpha_t(\bar\Q, \bar w)\\
&= \left\{u \in \bar M_t \; : \;  \begin{array}{rl} \bar w^\T  u \geq &\sum_{s=0}^{t-1} \esssup_{Z_s \in A_{R_s}} \bar w^\T  \bar\E_t^{\bar\Q}[-Z_s \ind{s}]\\ &+ \esssup_{Z \in A_t} \bar w^\T  \bar\E_t^{\bar\Q}[-Z \indn{\bbt_t}] \; \bar\P\mbox{-a.s.}\end{array}\right\}\\
&= \left\{u \in \bar M_t \; : \;  \sum_{s=0}^{t-1} \bar w_s^\T  u_s \ind{s} + \bar w_t^\T  u_t \indn{\bbt_t}
     \geq \sum_{s=0}^{t-1}  \esssup_{Z \in A_{R_s}} \bar w_s^\T (-Z_s) \ind{s}\right.\\
           &\qquad \left.+ \esssup_{Z \in A_t} \bar w_t^\T  \E_t^\Q[\sum_{s\in \bbt_t} \frac{1}{1 - \sum_{r = 0}^{t-1}\mu_r}\diag(\psi_s) (-Z_s)]\indn{\bbt_t} \; \bar\P\mbox{-a.s.}\right\}\\
&= \sum_{s=0}^{t-1} \left\{u_s \in M_s \; : \;  \bar w_s^\T  u_s \geq \esssup_{Z \in A_{R_s}} \bar w_s^\T (-Z_s) \right\} \ind{s}\\
   &\quad +  \left\{u_t \in M_t \; : \;  \bar w_t^\T  u_t \geq  \esssup_{Z \in A_t}\bar w_t^\T \E_t^\Q[\sum_{s\in \bbt_t}\frac{\diag(\psi_s)}{1 - \sum_{r = 0}^{t-1} \mu_r} (-Z_s)] \right\} \indn{\bbt_t}.
\end{align*}
Furthermore, by construction of $(\hat\Q,w) = W_t(\bar\Q,\bar w)$,
\begin{align*}
\sum_{s \in \bbt_t} w_s^\T \E_t^{\hat\Q_s}[Y_s] &= \bar w_t^\T \E_t^\Q[\sum_{s \in \bbt_t}\frac{1}{1 - \sum_{r = 0}^{t-1}\mu_r}\diag(\psi_s) Y_s]
 \end{align*}
for any $Y \in \rcal_t^{\infty,d}$.
Therefore, the desired result follows if $(\hat\Q,w) \in \W_t$.
In fact, this holds by the relation $w_t^s(\hat\Q_s,w_s) = \mu_s / (1 - \sum_{r = 0}^{t-1} \mu_r) \bar w_t^T(\bar\Q,\bar w)_s \in L_s^1(\bbr^d_+)$.

\item We will prove this result by first demonstrating that $\alpha_{R_s}$ is the minimal penalty function for $R_s$.  We will then consider the form of $\alpha_t$ as provided in the statement of the lemma; we accomplish this by proving that the appropriate representation holds and that $\bar W_t(\Q,w) \in \bar\W_t$ so that the representation is well-defined.
First, recall the definition of the minimal penalty function $\alpha_{R_s}$ for the risk measure $R_s: L_s^{\infty}(\bbr^d) \to \gcal(M_s;M_{s,+})$ as provided in Remark~\ref{rem:Rt}, i.e.,
\begin{align*}
\alpha_{R_s}(w) &= \bigcap_{Z \in A_{R_s}} \left\{(-Z + \Gamma_s(w)) \cap M_s\right\}\\
&= \left\{u \in M_s \; : \; w^\T u \geq \esssup_{Z \in A_{R_s}} w^\T (-Z) \right\}
\end{align*}
for any $w \in L_s^1(\bbr^d_+) \setminus M_s^\perp$.
Using the construction of $A_{R_s} = (\bar A_t)_s$ as given in Corollary~\ref{cor:equiv-primal-acceptance}, the construction of $\alpha_{R_s}$ in this lemma is trivial since $(\bar\P,w\ind{s}) \in \bar\W_t$ by inspection.

Second, recall the definition of the minimal penalty function $\alpha_t$ for the risk measure $\rho_t$ for processes as given in Corollary~\ref{cor:dual-process}, i.e.,
\begin{align*}
\alpha_t(\Q&,w) = \bigcap_{Z \in A_t} \sum_{s \in \bbt_t} \left(\E_t^{\Q_s}[-Z_s] + \Gamma_t(w_s)\right)\cap M_t\\
    &= \left\{u \in M_t \; : \;  \sum_{s\in \bbt_t} w_s^\T  u \geq \esssup_{Z \in A_t} \sum_{s\in \bbt_t} w_s^\T  \E_t^{\Q_s}[-Z_s] \right\}\\
    &= \left\{u \in M_t \; : \;  \begin{array}{l}\sum_{s\in \bbt_t} w_s^\T \diag(\ind{w_s > 0})  u \geq\\ \esssup_{Z \in A_t} \sum_{s\in \bbt_t} w_s^\T \diag(\ind{w_s > 0})  \E_t^{\Q_s}[-Z_s] \end{array}\right\}
\end{align*}
for any $(\Q,w) \in \W_t$.  The terminal equality above follows from the construction of the eligible assets and $w_s \geq 0 \; \Q_s\mbox{-a.s.}$.
By construction of $(\bar\Q,\bar w) = \bar W_t(\Q,w)$,
\[\left(\bar w^\T \bar\E_t^{\bar\Q}[Y]\right)_r = \sum_{s\in \bbt_t} w_s^\T  \E_t^{\Q_s}[Y_s]\]
for any $Y \in \bar L^{\infty}(\bbr^d)$ and $r \in \bbt_t$.
Therefore, using the construction of $A_t = \{X \in \rcal_t^{\infty,d} \; : \; X\ind{\bbt_t} \in \bar A_t\}$ as given in Corollary~\ref{cor:equiv-primal-acceptance}, the definition of $\alpha_t$ as provided in this lemma holds so long as $\bar\W_t(\Q,w) \in \bar\W_t$.
In fact, $\bar W_t(\Q,w) \in \bar\W_t$ holds by the relation 
\[\bar w_t^T(\bar W_t(\Q,w)) = (1 - \sum_{r = 0}^{t-1} \mu_r) \sum_{s \in \bbt_t} \mu_s^{-1} w_t^s(\Q_s,w_s)\ind{s} \in \bar L^1(\bbr^d_+).\]
\end{enumerate}
\end{proof}

\begin{remark}
We highlight that, in the full eligible space setting $M = \bbr^d$ with the notation given in Corollary~\ref{cor:equiv-primal-full}, the penalty function representation simplifies insofar as $\alpha_{R_s}(\bar w_s) = \Gamma_s(\bar w_s)$ if $\bar w_s \in C_s^{\ast}$ and $\emptyset$ otherwise.
\end{remark}

We, additionally, consider the relation between the dual representations for conditionally coherent risk measures.  This is provided in the following corollary utilizing the maximal set of dual variables.
\begin{corollary}
Fix time $t \in \bbt$.
\begin{enumerate}
\item Let $\rho_t: \rcal_t^{\infty,d} \to \gcal(M_t;M_{t,+})$ be a closed conditionally coherent risk measure for processes and let $R_s: L_s^{\infty}(\bbr^d) \to \gcal(M_s;M_{s,+})$, $s < t$, be a series of closed conditionally coherent risk measures.  The associated closed conditionally coherent risk measure on the optional filtration $\bar R_t: \bar L^{\infty}(\bbr^d) \to \gcal(\bar M_t;\bar M_{t,+})$ defined in~\eqref{eq:primal-barR} has dual representation w.r.t.\ the set of dual variables
    \begin{align*}
    \bar\W_t^{\max} = \{(\bar\Q,\bar w) \in \bar\W_t \; : \; \bar w_s \in \W_{R_s}^{\max} \; \forall s = 0,\dots,t-1, \; W_t(\bar\Q,\bar w) \in \W_t^{\max}\}.
    \end{align*}
\item Consider a closed conditionally coherent risk measure for vectors on the optional filtration $\bar R_t: \bar L^{\infty}(\bbr^d) \to \gcal(\bar M_t;\bar M_{t,+})$.  The associated closed conditionally coherent risk measure for processes $\rho_t: \rcal_t^{\infty,d} \to \gcal(M_t;M_{t,+})$ and series of closed conditionally coherent risk measures $R_s: L_s^{\infty}(\bbr^d) \to \gcal(M_s;M_{s,+})$ defined in~\eqref{eq:primal-rho} and~\eqref{eq:primal-R}, respectively, have dual representations w.r.t.\ the set of dual variables
    \begin{align*}
    \W_{R_s}^{\max} &= \bigcup_{\bar\Q \in \bar\M_t(\bar\P)^d} \{\bar w_s \; : \; (\bar\Q,\bar w) \in \bar\W_t^{\max}\},\\
    \W_t^{\max} &= \{(\Q,w) \in \W_t \; : \; \bar W_t(\Q,w) \in \bar\W_t^{\max}\}.
    \end{align*}
\end{enumerate}
\end{corollary}
\begin{proof}
This is an immediate consequence of Lemma~\ref{lemma:equiv-dual} utilizing the indicator notion of the penalty functions, i.e.,
\begin{align*}
\alpha_t(\Q,w) &= \begin{cases} \sum_{s \in \bbt_t} \Gamma_t(w_s) \cap M_t &\text{if } (\Q,w) \in \W_t^{\max} \\ \emptyset &\text{else,}\end{cases} \\
\alpha_{R_s}(w_s) &= \begin{cases} \Gamma_s(w_s) \cap M_s &\text{if } w_s \in \W_{R_s}^{\max} \\ \emptyset &\text{else,}\end{cases} \\
\bar\alpha_t(\bar\Q,\bar w) &= \begin{cases} \bar\Gamma_t(\bar w) \cap \bar M_t &\text{if } (\bar\Q,\bar w) \in \bar\W_t^{\max} \\ \emptyset &\text{else.}\end{cases}
\end{align*} 
\end{proof}

We now use the results of this section in order to prove the dual representation for risk measures for processes as presented in Corollary~\ref{cor:dual-process}. Specifically, we take advantage of the known dual representation for set-valued risk measures for vectors (Corollary~\ref{cor:dual-vector}) and the equivalence of forms for the set-valued risk measures for processes and vectors presented above in order to prove this new dual representation for risk measures for processes.
\begin{proof}[Proof of Corollary~\ref{cor:dual-process}]
Let $\rho_t: \rcal_t^{\infty,d} \to \gcal(M_t;M_{t,+})$ be some closed conditionally convex risk measure for processes.  Additionally, fix a sequence of closed and conditionally convex risk measures $R_s: L_s^{\infty}(\bbr^d) \to \gcal(M_s;M_{s,+})$ for $s < t$.  The existence and uniqueness of a closed conditionally convex risk measure on the optional filtration $\bar R_t: \bar L^{\infty}(\bbr^d) \to \gcal(\bar M_t;\bar M_{t,+})$ such that $\rho_t(X) = \bar R_t(X \indn{\bbt_t})_t$ for any $X \in \rcal_t^{\infty,d}$ is guaranteed by Theorem~\ref{thm:equiv-primal}.
Consider, also, the decomposition of the penalty function from Lemma~\ref{lemma:equiv-dual} to determine that $\alpha_t(\Q,w) = \bar\alpha_t(\bar\Q,\bar w)_t$ for specifically constructed $(\bar\Q,\bar w) = \bar W_t(\Q,w)$.
Utilizing the dual representation for risk measures for vectors as provided in Corollary~\ref{cor:dual-vector}:
\begin{align*}
\rho_t(X) &= \bigg(\bigcap_{(\bar\Q,\bar w) \in \bar\W_t} \Big(\big(\bar\E_t^{\bar\Q}[-X\indn{\bbt_t}] + \bar\Gamma_t(\bar w)\big)\cap \bar M_t -^{\sbullet} \bar\alpha_t(\bar\Q,\bar w)\Big)\bigg)_t\\
&= \bigcap_{(\Q\otimes\psi,\bar w) \in \bar\W_t} \left(\begin{array}{l}\big(\E_t^\Q\left[-\sum_{s \in \bbt_t} \frac{\psi_s}{1 - \sum_{r = 0}^{t-1} \psi_r} X_s\right] + \Gamma_t(\bar w_t)\big) \cap M_t\\ \quad -^{\sbullet} \bar\alpha_t(\Q\otimes\psi,\bar w)_t\end{array}\right)\\
&= \bigcap_{(\Q\otimes\psi,\bar w) \in \bar\W_t} \bigg(\sum_{s \in \bbt_t}\Big(\big(\E_t^{\hat\Q_s}[-X_s] + \Gamma_t(w_s)\big) \cap M_t\Big) -^{\sbullet} \alpha_t(\hat\Q,w)\bigg)
\end{align*}
where $(\hat\Q,w) = W_t(\bar\Q,\bar w)$ is defined as in Lemma~\ref{lemma:equiv-dual}\eqref{lemma:equiv-dual-1}.
To recover~\eqref{eq:dual-process}, it remains to show that \[\{(w_t^s(\Q_s,w_s))_{s \in \bbt_t} \; : \; (\Q,w) \in \W_t\} = \{(w_t^s(W_t(\bar\Q,\bar w)_s))_{s \in \bbt_t} \; : \; (\bar\Q,\bar w) \in \bar\W_t\}.\]
First, by construction, $\{W_t(\bar\Q,\bar w) \; : \; (\bar\Q,\bar w) \in \bar\W_t\} \subseteq \W_t$ which implies $\supseteq$ in the desired equality.
Second, for any $(\Q,w) \in \W_t$ there exists $\bar W_t(\Q,w) \in \bar\W_t$ (as constructed in Lemma~\ref{lemma:equiv-dual}\eqref{lemma:equiv-dual-2}) so that $w_t^s(\Q_s,w_s) = w_t^s(W_t(\bar W_t(\Q,w))_s)$ for any $s \in \bbt_t$ which implies $\subseteq$ in the desired equality.
Thus the proof in the convex case is complete.

For the conditionally coherent case, define
\begin{align*}
\beta_t(\Q, w)
=\begin{cases}
\sum\limits_{s \in \bbt_t} \Gamma_t(w_s) \cap M_t  &\text{if } \essinf\limits_{Z \in A_t} \sum_{s\in \bbt_t} w_s^\T \E^{\Q_s}_t[Z_s]=0 ~ \P\mbox{-a.s.}\\
\emptyset &\text{else.}
\end{cases}
\end{align*}
Then from the positive homogeneity of $\rho$, we have that  $\beta_t(\Q, w)\subseteq \alpha_t(\Q, w)$, which yields that
\begin{align*}
&\sum_{s \in \bbt_t} \Big(\big(\E_t^{\Q_s}[-X_s] + \Gamma_t(w_s)\big) \cap M_t\Big) -^{\sbullet} \alpha_t(\Q, w)\\ &\qquad \subseteq \sum_{s \in \bbt_t} \Big(\big(\E_t^{\Q_s}[-X_s] + \Gamma_t(w_s)\big) \cap M_t\Big) -^{\sbullet} \beta_t(\Q, w).
\end{align*}
Note that
\begin{align*}
&\sum_{s \in \bbt_t} \left(\E_t^{\Q_s}[-X_s] + \Gamma_t(w_s)\right) \cap M_t -^{\sbullet}  (\sum_{s\in \bbt_t} \Gamma_t(w_s)) \cap M_t\\ &\qquad = \sum_{s \in \bbt_t} \left(\E_t^{\Q_s}[-X_s] + \Gamma_t(w_s)\right) \cap M_t,\\
&\sum_{s \in \bbt_t} \left(\E_t^{\Q_s}[-X_s] + \Gamma_t(w_s)\right) \cap M_t -^{\sbullet} \emptyset = M_t,
\end{align*}
we obtain that
\begin{align*}
\rho_t(X) \subseteq \bigcap_{(\Q,w) \in \W_t^{\max}} \sum_{s \in \bbt_t} \left(\E_t^{\Q_s}[-X_s] + \Gamma_t(w_s)\right) \cap M_t.
\end{align*}
Finally, from \eqref{eq:dual-process}, it follows that $u_t \in \rho_t(X)$ if and only if
    \[\essinf\limits_{(\Q, w) \in \W_t} \left\{\sum\limits_{s\in \bbt_t} w_s^\T  u_t
                          + \esssup\limits_{Z \in A_t} \sum\limits_{s\in \bbt_t} w_s^\T \E^{\Q_s}_t[-Z_s]
                               + \sum\limits_{s\in \bbt_t} w_s^\T \E^{\Q_s}_t[X_s] \right\}\geq 0.\]
    Assume $u \in \bigcap_{(\Q,w) \in \W_t^{\max}} \sum_{s \in \bbt_t} (\E_t^{\Q_s}[-X_s] + \Gamma_t(w_s)) \cap M_t$.  Then, immediately,
    \[\essinf\limits_{(\Q, w) \in \W_t^{\max}}  \sum\limits_{s\in \bbt_t} w_s^\T  u
                               + \esssup\limits_{Z \in A_t} \sum\limits_{s\in \bbt_t} w_s^\T \E^{\Q_s}_t[-Z_s] + \sum\limits_{s\in \bbt_t} w_s^\T \E^{\Q_s}_t[X_s] \geq 0.\]
    Additionally and clearly,
    \[\essinf\limits_{(\Q, w) \in \W_t \setminus \W_t^{\max}}\left\{\sum\limits_{s\in \bbt_t} w_s^\T  u
                              + \esssup\limits_{Z \in A_t} \sum\limits_{s\in \bbt_t} w_s^\T \E^{\Q_s}_t[-Z_s] + \sum\limits_{s\in \bbt_t} w_s^\T \E^{\Q_s}_t[X_s]\right\}\geq 0.\]
    Hence, the essential infimum taken over the full set of dual variables $\W_t$ must also be bounded from below by $0$,
    which implies that $u \in \rho_t(X)$ and the equivalence is shown. 
\end{proof}

\section{Equivalence of multiportfolio time consistency}\label{sec:mptc}

Though multiportfolio time consistency of set-valued risk measures for processes was presented in~\cite{CH17},
in this work we find that a slight variation of that definition is useful for this work.  In particular, motivated by Theorem~\ref{thm:equiv-primal}, we consider a joint definition of both the risk measure for processes $\rho_t$ and the series of restricted conditional risk measures for vectors $R_s$.
\begin{definition}\label{defn:mptc-joint}
The pair $(\rho_t, (R_s^t)_{s=0}^{t-1})_{t\in\bbt}$ is \textbf{\emph{jointly} multiportfolio time consistent} if:
\begin{enumerate}
\item $\rho$ is multiportfolio time consistent as in Definition~\ref{defn:mptc-process};
\item\label{item:mptc-joint-2} for any times $t < s$, $X_{r} \in L_{r}^{\infty}(\bbr^d)$ and $B_{r} \subseteq L_{r}^{\infty}(\bbr^d)$ for $r \in [t,s)$, and $Z \in \rcal_s^{\infty,d}$:
    \begin{align*}
    &R_r^s(X_r) \subseteq \bigcup_{Y_r \in B_r} R_{r}^s(Y_r) \quad \forall r \in [t,s) \\
    &\Rightarrow \; \rho_t(\sum_{r = t}^{s-1} X_r\ind{r} + Z\indn{\bbt_s}) \subseteq \bigcup_{Y_t \in B_t} \cdots \bigcup_{Y_{s-1} \in B_{s-1}} \rho_t(\sum_{r = t}^{s-1} Y_r\ind{r} + Z\indn{\bbt_s});
    \end{align*}
\item\label{item:mptc-joint-3} for any times $r < t < s$, $X_{r} \in L_{r}^{\infty}(\bbr^d)$ and $B_{r} \subseteq L_{r}^{\infty}(\bbr^d)$:
    \[R_r^s(X_r) \subseteq \bigcup_{Y_r \in B_r} R_r^s(Y_r) \; \Rightarrow \; R_r^t(X_r) \subseteq \bigcup_{Y_r \in B_r} R_r^t(Y_r).\]
\end{enumerate}
\end{definition}

Conceptually, $(\rho,R)$ is jointly multiportfolio time consistent if, in addition to multiportfolio time consistency it satisfies two additional consistency properties. \eqref{item:mptc-joint-2} If two claims are identical after time $s \in \bbt$ and the restricted risk measures (associated with $\rho_s$) on the claims at $r \in [t,s)$ are ordered as multiportfolio time consistency, then the claims as measured by $\rho_t$ at time $t$ must also satisfy the ordering of risks; i.e., the restricted risk measures are consistent in time with the risk measures for processes. And \eqref{item:mptc-joint-3} the ordering of portfolios induced by the restricted risk measure at time $r < t < s \in \bbt$ associated with $\rho_s$ should hold for the restricted risk measure at time $r$ associated with $\rho_t$ as well.  This trivially holds if the restricted risk measure at time $r$ is independent of the risk measure for processes with which it is associated, i.e., $R_r^t = R_r^s$ for $r < t < s$; such a setting is taken in Remark~\ref{rem:mptc-full} below.

We now turn to the final main result of this work -- the equivalence of multiportfolio time consistency for set-valued risk measures for processes and vectors.  This is, again, akin to the results presented for scalar risk measures in~\cite{acciaio2012risk}.  Fundamentally, as the following proof is demonstrated for the definition of multiportfolio time consistency, these results can be used to construct the various equivalent formulations as well; in particular, we highlight the recursive formulation
presented in~\cite{FR12,CH17} and the cocycle condition for penalty functions
 in~\cite{FR12b}.
\begin{theorem}\label{thm:equiv-mptc}
\begin{enumerate}
\item\label{thm:equiv-mptc-1} Consider a jointly multiportfolio time consistent pair of risk measures $(\rho, R)$.
The associated time decomposable risk measure on the optional filtration $\bar R_t: \bar L^{\infty}(\bbr^d) \to \ucal(\bar M_t;\bar M_{t,+})$ defined in~\eqref{eq:primal-barR} is multiportfolio time consistent.
\item\label{thm:equiv-mptc-2} Consider a multiportfolio time consistent, time decomposable risk measure for random vectors on the optional filtration $\bar R$.
The associated risk measure for processes $\rho_t: \rcal_t^{\infty,d} \to \ucal(M_t;M_{t,+})$ and series of risk measures $R_s^t: L_s^{\infty}(\bbr^d) \to \ucal(M_s;M_{s,+})$, $s=0, \dots, t-1$ for $t=0, \dots, T$, defined in~\eqref{eq:primal-rho} and~\eqref{eq:primal-R} are jointly multiportfolio time consistent.
\end{enumerate}
\end{theorem}
\begin{proof}
Throughout we will fix times $t < s$.
\begin{enumerate}
\item Let $B := \{(b_0,b_1,\dots,b_{s-1},b_s) \; : \; b_r \in B_r, r < s, \; b_s \in B_s\}$ for any arbitrary sequence of sets $B_r \subseteq L_r^{\infty}(\bbr^d)$ for $r < s$ and $B_s \subseteq \rcal^{\infty,d}_s$ and $X \in \bar L^{\infty}(\bbr^d)$ such that:
\[\bar R_s(X) \subseteq \bigcup_{Y \in B} \bar R_s(Y).\]
    By construction, this implies $R_r^s(X_r) \subseteq \bigcup_{Y_r \in B_r} R_r^s(Y_r)$ for every $r < s$ and $\rho_s(X\indn{\bbt_s}) \subseteq \bigcup_{Y_s \in B_s} \rho_s(Y_s\indn{\bbt_s})$.
    Immediately, from the third condition of joint multiportfolio time consistency, we know $R_r^t(X_r) \subseteq \bigcup_{Y_r \in B_r} R_r^t(Y_r)$ for every $r < t$,
    it remains to show that $\rho_t(X\indn{\bbt_t}) \subseteq \bigcup_{Y \in B} \rho_t(Y\indn{\bbt_t})$:
\begin{align*}
\rho_t(X\indn{\bbt_t}) &\subseteq \bigcup_{Y_t \in B_t} \cdots \bigcup_{Y_{s-1} \in B_{s-1}} \rho_t(\sum_{r = t}^{s-1} Y_r\ind{r} + X\indn{\bbt_s}) \\
&\subseteq \bigcup_{Y_t \in B_t} \cdots \bigcup_{Y_{s-1} \in B_{s-1}} \bigcup_{Y_s \in B_s} \rho_t(\sum_{r = t}^{s-1} Y_r\ind{r} + Y_s\indn{\bbt_s}) = \bigcup_{Y \in B} \rho_t(Y\indn{\bbt_t}),
\end{align*}
where the first inclusion follows from the second condition of joint multiportfolio time consistency and the second inclusion from the first condition of joint multiportfolio time consistency.  Therefore, by construction of $\bar R_t$ it immediately follows that $\bar R_t(X) \subseteq \bigcup_{Y \in B} \bar R_t(Y)$.


\item Note, first, since $(\bar R_t)_{t\in\bbt}$ is time decomposable, from Theorem~\ref{thm:equiv-primal}, we have that
    \[\bar R_t(X) = \sum_{r = 0}^{t-1} R_r^t(X_r)\ind{r} + \rho_t(\pi_{t,T}(X))\indn{\bbt_t}\]
    for any time $t \in \bbt$. For this proof, we will only directly demonstrate the first condition of joint multiportfolio time consistency; the latter two properties follow from identical arguments.
    Fix $X \in \rcal^{\infty,d}$ and $B \subseteq \rcal^{\infty,d}$ such that
    \begin{align*}
    \rho_s(\pi_{s,T}(X)) \subseteq  \bigcup_{Y \in B} \rho_s(\pi_{s,T}(Y)).
    \end{align*}
    This implies, by the decomposition of $\bar R_s$, that
    \[\bar R_s(Z\ind{[t,s)} + X\indn{\bbt_s}) \subseteq \bigcup_{Y \in B} \bar R_s(Z\ind{[t,s)} + Y\indn{\bbt_s})\]
    for any $Z \in \rcal^{\infty,d}$.
    Utilizing multiportfolio time consistency of $\bar R$, we recover
    \[\bar R_t(Z\ind{[t,s)} + X\indn{\bbt_s}) \subseteq \bigcup_{Y \in B} \bar R_t(Z\ind{[t,s)} + Y\indn{\bbt_s})\]
    for any $Z \in \rcal^{\infty,d}$.
    By the decomposition of $\bar R_t$, we immediately conclude
    \[\rho_t(Z\ind{[t,s)} + X\indn{\bbt_s}) \subseteq \bigcup_{Y \in B} \rho_t(Z\ind{[t,s)} + Y\indn{\bbt_s})\]
    for any $Z \in \rcal^{\infty,d}$, i.e., $\rho$ is multiportfolio time consistent.
 \end{enumerate}
\end{proof}

From Theorems \ref{thm:equiv-primal} and \ref{thm:equiv-mptc} and using results in~\cite{FR12b,FR15-supermtg}, one can steadily deduce some equivalent characterizations of multiportfolio time consistency for set-valued dynamic risk measures for processes, such as the cocycle condition on the sum of minimal penalty functions and supermartingale relation. 

\begin{remark}\label{rem:mptc-full}
Consider the setting in which all assets are eligible, i.e., assume that $M = \bbr^d$.
Let $\rho$ be a multiportfolio time consistent risk measure for processes (see Definition~\ref{defn:mptc-process})
and define $R_s: L_s^{\infty}(\bbr^d) \to \ucal(L_s^{\infty}(\bbr^d);L_s^{\infty}(\bbr^d_+))$ by
\begin{align*}
R_s(X_s) := -X_s + L_s^{\infty}(\bbr^d_+), \quad X_s \in L_s^{\infty}(\bbr^d).
\end{align*}
By monotonicity, $(\rho, R)$ is jointly multiportfolio time consistent. Therefore, by Theorem~\ref{thm:equiv-mptc}, $\rho$ is multiportfolio time consistent if and only if the associated risk measure for vectors $\bar R$ is multiportfolio time consistent.
\end{remark}

We conclude this discussion by remarking that, in this discrete time setting, it is sufficient to consider multiportfolio time consistency defined for single time steps only through a sequential application of the recursive definition (i.e., setting $s = t+1$).  This is discussed in, e.g.,~\cite{FR12b} as well.

\section{Conclusion}\label{sec:conclusion}
In this work we demonstrated the equivalence between set-valued risk measures for processes with those for random vectors on the optional filtration.  Such considerations allow for the application of results from one stream of literature to the other.  In particular, we highlight that the set-valued risk measures for vectors is a more mature field with more results on dual representations and equivalent forms for multiportfolio time consistency.  Herein, we used this equivalence to prove a new dual representation for risk measures for processes.  We also highlight that the equivalence of these risk measures allows for the generalization of, e.g., the cocycle condition on the sum of minimal penalty functions or the supermartingale relation for multiportfolio time consistency for risk measures for vectors to risk measures for processes.  We caution the reader that such extensions, while following the results of this work, will require the use of a series of conditional risk measures for vectors on the original filtration in addition to the risk measure for processes which is \emph{not} necessary if considering scalar risk measures as presented in~\cite{acciaio2012risk}.

\section*{Acknowledgements}
The authors are very grateful to the Editors and the anonymous referees for their comments and suggestions which led to the present greatly improved version of the manuscript.
%

\bibliographystyle{plain}
\bibliography{biblio}

\end{document}